\documentclass[sigconf]{acmart}

\usepackage{booktabs} 
\usepackage{graphicx}
\usepackage{float}
\usepackage{amssymb}
\usepackage{bm}
\usepackage{epsfig}
\usepackage{amsthm}
\usepackage{multirow}
\usepackage{url}
\usepackage{mathrsfs}
\usepackage{breqn}
\usepackage{bm}
\usepackage[tight,footnotesize]{subfigure}
\usepackage{flushend}
\usepackage{diagbox}
\usepackage{multirow}
\usepackage{setspace}
\usepackage{textcomp}
\usepackage{booktabs}
\usepackage{enumerate}
\usepackage{array}
\usepackage{amsmath}
\makeatletter
\newif\if@restonecol
\makeatother

\usepackage[linesnumbered,ruled,vlined]{algorithm2e}
\usepackage{algpseudocode}

\newtheorem{theorem}{\noindent\textbf{Theorem}}

\newtheorem{proposition}[theorem]{\noindent\textbf{Proposition}}

\newenvironment{sequation}{\begin{equation}\small\setlength{\abovedisplayskip}{3pt}\setlength{\belowdisplayskip}{3pt}}{\end{equation}}
\newenvironment{sequation*}{\begin{equation*}\small\setlength{\abovedisplayskip}{3pt}\setlength{\belowdisplayskip}{3pt}}{\end{equation*}}

\renewcommand\footnotetextcopyrightpermission[1]{} 
\setcopyright{none}

\begin{document}

\title{PingAn: An Insurance Scheme for Job Acceleration in Geo-distributed Big Data Analytics System}

\author{Tiantian~Wang, Zhuzhong~Qian, Sanglu~Lu}
\affiliation{State~Key~Laboratory~for~Novel~Software~Technology,~Nanjing~University,~Nanjing,~China}
\email{Email:  dz1633012@smail.nju.edu.cn, qzz@nju.edu.cn, sanglu@nju.edu.cn}



%


\begin{abstract}
Geo-distributed data analysis in a cloud-edge system is emerging as a daily demand.
Out of saving time in wide area data transfer, some tasks are dispersed to the edges.
However, due to limited computing, overload interference and cluster-level unreachable troubles, efficient execution in the edges is hard, which obstructs the guarantee on the efficiency and reliability of jobs.
Launching copies across clusters can be an insurance on a task's completion.
Considering cluster heterogeneity and accompanying remote data fetch, cluster selection of copies affects execution quality, as different insuring plans drive different revenues.
For providing On-Line-Real-Time analysis results, a system needs to insure the geo-distributed resource for the arriving jobs.
Our challenge is to achieve the optimal revenue by dynamically weighing the gains due to insurance against the loss of occupying extra resource for insuring.

To this end, we design PingAn, an online insurance algorithm promising $(1+\varepsilon)\!\!-\!speed \  o(\frac{1}{\varepsilon^2+\varepsilon})\!\!-\!competitive$ in sum of the job flowtimes via cross-cluster copying for tasks.
PingAn shares resource among the anterior fraction of jobs with the least unprocessed datasize and the fraction is adjustable to fit the system load condition.
After sharing, PingAn concretely insures for tasks following efficiency-first reliability-aware principle to optimize the revenue of copies on jobs' performance.
Trace-driven simulations demonstrate that PingAn can reduce the average job flowtimes by at least $14\%$ than the state-of-the-art speculation mechanisms.
We also build PingAn in Spark on Yarn System to verify its practicality and generality. Experiments show that PingAn can reduce the average job flowtimes by up to $40\%$ comparing to the default Spark execution.
\end{abstract}

%
%
%


\maketitle

\section{Introduction}
Today, several cloud applications are moving some of their functionality to edge devices to improve user-perceived fluency of interactions.
The edges are considered as the extension of traditional data centers and they together constitute a large scale cloud-edge system ~\cite{Tan2017Online,Mao2017A,Tong2016A,Chen2017Joint,Calder2013Mapping}.
A large amount of user data, e.g. logs, transaction records and traces, is generated and stored in edges.
Analyses on these geo-distributed user data precipitate many realtime commerce-crucial decisions, for instance, user behavior predictions, load balancing and attack detections ~\cite{Gupta2014Mesa,Morita2014Aggregation,hsieh2017gaia}.

Traditional centralized data analysis needs to transfer all required data to one site, which is time consuming. Thus, modern data analytic platforms tend to disperse tasks to edges and process data locally. Previous works ~\cite{kloudas2015pixida,Vulimiri2015Global,Hung2015Scheduling,Pu2015Low,viswanathan2016clarinet,wang2017lube} carefully design the tasks scheduling policy to minimize the costly WAN data transmission, so as to speed up job completion.

However, achieving efficient geo-distributed data analyses still encounters obstacles besides the WAN bandwidth limitation.
Owing to the limited resources, such as computing slots and import-export bandwidth, edge clusters may be easily overloaded under dynamic user access patterns, or even suffer a cluster-level unreachable trouble.
Thus, tasks in edges sometimes perform badly even fail.
In addition, edge clusters are heterogeneous.
One task running on different edges may have totally different execution quality, while it is very similar among data centers.

One potential method to avoid the unpredictability of edges is utilizing idle resources to clone some critical tasks in multiple edges or data centers to guarantee the job completion time.
Actually, several intra-cluster data analytic platforms have already adopted similar idea to handle straggling tasks ~\cite{Zaharia2008Improving,Dean2010MapReduce,Ananthanarayanan2010Reining,Ren2015Hopper,Chaiken2008SCOPE,Ananthanarayanan2013Effective,Ananthanarayanan2014GRASS,Xu2015Task}, but these cluster-scale straggler-handling mechanisms are unsuited for the inter-cluster insurance in a cloud-edge system.

Firstly, the normal task execution in a cluster assumes slot-independent, i.e., for a task, the execution time on each normal slot is similar, while it is quite different among different edges.
Secondly, the difference of data fetch time caused by the task location is almost imperceptible inside cluster because that the data always has copies inside cluster and intra-cluster bandwidth is abundant, e.g. HDFS has three copies by default.
In contrast, cloning a task in another edge or data center incurs inter-cluster data transmission over scarce WAN bandwidth, thus, the difference of data fetch among inter-cluster copies is non-negligible.
In order to effectively speed up the jobs, we need to consider the impact of cluster heterogeneity and remote data fetch when insuring.

The revenue of insuring is embodied in the improvement of the task's expected execution speed (efficiency) and the probability of completion (reliability).
To measure efficiency improvement, we capture the heterogeneous performance of geo-distribution resource from the recent execution logs and quantify the effect of a task's insurance plan as the change of its execution speed.
For reliability, we quantify the effect of insuring as the increase of task's completion probability and utilize the inter-cluster copies to maximally avoid failure caused by the cluster-level unreachable troubles.

Based on these quantifications, in this paper, we design PingAn, an online fine-grained insurance algorithm aiming to minimize the sum of job flowtimes.
First, PingAn permits the first $\varepsilon$ fraction of jobs with the least unprocessed data size to share the computing resource.
Considering that for a task, the marginal revenue of an extra copy decreases as the task's copy number increases, tuning $\varepsilon$ to accommodate the system's load condition is expected to motivate the best effect of copies under limited resource.
Some hints about $\varepsilon$ selection are also given in the paper according to the experiment results.

When concretely insuring for tasks, we care for both efficiency and reliability on copy's cluster selection.
However, towards our aim, trading off the gains of efficiency against the loss of ignoring reliability is hard and vice versa.
Thus, irresolution arises in the course of insuring, such as at the moment of selecting cluster for a copy, arranging copies for a task, deciding the insuring order of tasks in a job and disposing the collision of preferential clusters among concurrent jobs.

PingAn insures for tasks adhering to the efficient-first reliability-aware principle which relys on the factor that the cluster-level trouble occasionally occurs but seriously harms a wave of jobs' performance.
Further, PingAn improves the efficiency via confining the worst execution rate for each task and averting a worse usage for each slot.

We prove that our online insurance algorithm, PingAn, is a competitive online algorithm in theory and verify the improvement effect of PingAn via trace-driven simulations.
Further, we develop PingAn in Spark on Yarn system to handle real-world workloads due to its practicality and generality.
To be practical, PingAn works without any priori knowledge of jobs beyond the current job progress.
To be general, PingAn serves for general geo-distributed data analysis jobs with any precedence constraints among tasks.

To summarize, we make three main contributions:
\begin{itemize}
\item
We model the dynamic performance of geo-distributed resource to quantify its impact on task completion and formulate our online insurance problem as an optimization problem devoting to make an insurance plan to minimize the sum of job flowtimes.
\item
We design PingAn, an online insurance algorithm and prove it is $o(1+\varepsilon)\!\!-\!\!speed$ $o(\frac{1}{\varepsilon^2+\varepsilon})\!-\!competitive$ in the sum of job flowtimes where $0<\varepsilon<1$.
In simulations, we demonstrate that PingAn can drastically improve the job performance in a cloud-edge environment under any system load condition and reduce the average job flowtimes by at least 14$\%$ than the best speculation mechanisms under heavy load and the improvement is up to 62$\%$ under lighter load.
\item
We develop a prototype of PingAn in Spark on Yarn system and run jobs in comparison with the default Spark executions.
Experiments show that PingAn can guarantee the efficient and reliable job executions in real-world implementation.
\end{itemize}
\section{Related Work}
\textbf{Geo-distributed data analyses:}
Works ~\cite{kloudas2015pixida,Vulimiri2015Global,Hung2015Scheduling,Pu2015Low,viswanathan2016clarinet,Hu2016Flutter} as pioneers devote to the performance problem of geo-distributed data analysis.
Iridium ~\cite{Hung2015Scheduling} coordinates data and task placement to improve query response.
Clarinet ~\cite{Pu2015Low} makes query execution plans with a wide-area network awareness.
Flutter ~\cite{Hu2016Flutter} minimizes the completion time of stage via optimizing task assignment across data centers .
These proposed solutions reduce WAN transfer to improve job performance and assumes unlimited computing resources in data centers at all times.
However, in a cloud-edge scenario, the resource-limited and unreliable edge clusters harms the job performance.
We consider the edges' limitation and ensure the task execution via inter-cluster task copying.

\textbf{Passive detection speculation mechanism:}
This part of works mitigate the abnormal task impact on job completion via monitoring tasks' execution and restart a new copy for identified straggler.
Initially, Google MapReduce system speculatively schedules copies for the remaining tasks at the end of job ~\cite{Dean2010MapReduce}.
It restrains the long-tail tasks but wastes resource on lots of normal tasks. Thus LATE ~\cite{Zaharia2008Improving} and its extended works identifies slow tasks accurately via delicately comparing tasks' progress rate.
Mantri ~\cite{Ananthanarayanan2010Reining} schedules a copy for a task only when the task's total resource consumption decreases.
Hopper ~\cite{Ren2015Hopper} designs the best speculation-aware job scheduler under its task duration model.

However, the above cluster-scale speculation mechanisms lose efficacy in cloud-edge environment.
First, monitoring lots of remote tasks is costly.
Further, the cluster-level unreachable troubles, for instance, power supply interruption, master server shutting down, the failure of high layer exchanger which leads to a network disconnection and many more complicated cases caused by a series of operation accidents, obstruct the system to timely detect straggling tasks.
Then, the normal task standard is indecisive since the cluster heterogeneity, which delays speculation.
Besides, the time-consuming WAN transfer in a restart copy further destroy the acceleration effect of speculation.
We insuring for tasks at the start of execution to avoid these problems.

\textbf{Proactive clone mechanism:}
This part of works devote to reduce the straggler occurrence of a job via task cloning at start.
Dolly ~\cite{Ananthanarayanan2013Effective} refines the straggler-occurring likelihood of some jobs beyond a certain threshold.
~\cite{Xu2015Task} adopts task cloning to speed up job completion and proposes an competitive online scheduling algorithm to optimize the sum of the job flowtimes.
Work in ~\cite{Xu2015Optimization} proposes Smart Cloning Algorithm to maximize the sum of job utility via task cloning.
Given the cluster heterogeneity and WAN transfer demand, the copy execution in different clusters differs.
Thus, in a cloud-edge system, the above cluster-scale cloning mechanisms which only decide the copy number for each task fail to achieve the effect of copies on job performance improvement.
To this end, we make the fine-grained insurance plan to optimize copy effect.
\section{System Overview and Insurance Problem}
\subsection{Geo-distributed Data Analysis System with PingAn}
PingAn utilizes users' geo-distributed resource to guarantee the efficiency of their routine data analyses.
We develop PingAn in Spark ~\cite{zaharia2012resilient} on Yarn ~\cite{vavilapalli2013apache} system across multiple clusters as shown in Figure \ref{fig:PingAn_architecture}.
The resource of each cluster (dash line in Figure \ref{fig:PingAn_architecture}) is managed via one {\tt ResourceManager(RM)} in Yarn.
{\tt RM} receives jobs from Spark client and resolves the job's description to generate a corresponding {\tt AppMaster(AM)} for each job.
Inside {\tt AM}, {\tt DAGScheduler} creates {\tt TaskSet} for the ready tasks.
PingAn works as shown in Figure \ref{fig:p1}.
\begin{enumerate}[a)]
\item
{\tt DAGScheduler} fetches the data location information of tasks from {\tt OutputRecorder} in {\tt AM} and inserts it into {\tt TaskSet}.
The {\tt OutputRecorder} records the intermediate data location once a completed task reports its output message to it.
\item
{\tt TaskSet} is then send to PingAn and waits in {\tt TaskSetPool}.
Multiple {\tt TaskSet}s in {\tt TaskSetPool} are queued in an ascending order of unprocessed data size.
\item
{\tt PerformanceModeler(PM)} in PingAn regularly collects the execution information in each cluster from {\tt RM}s and models the dynamic of resource capacity.
\item
{\tt Insurancer} in PingAn periodically fetches the {\tt TaskSet}s from {\tt TaskSetPool} and draws up an insurance plan for each task with an aware of its completion time in each cluster which is estimated depending on the resource performance model in {\tt PM}.
\item
{\tt TaskSet} along with its insurance plan is sent back to {\tt TaskScheduler} for execution.
{\tt AM} sends the resource(container) requests to {\tt RM}s in the clusters specified in the insurance plan and launched the tasks on the obtained containers.
\end{enumerate}
\begin{figure}[htb]
  \centering
  \subfigure[The cooperation between Spark on Yarn and PingAn]{
    \label{fig:p1}
    \includegraphics[width=0.8\linewidth]{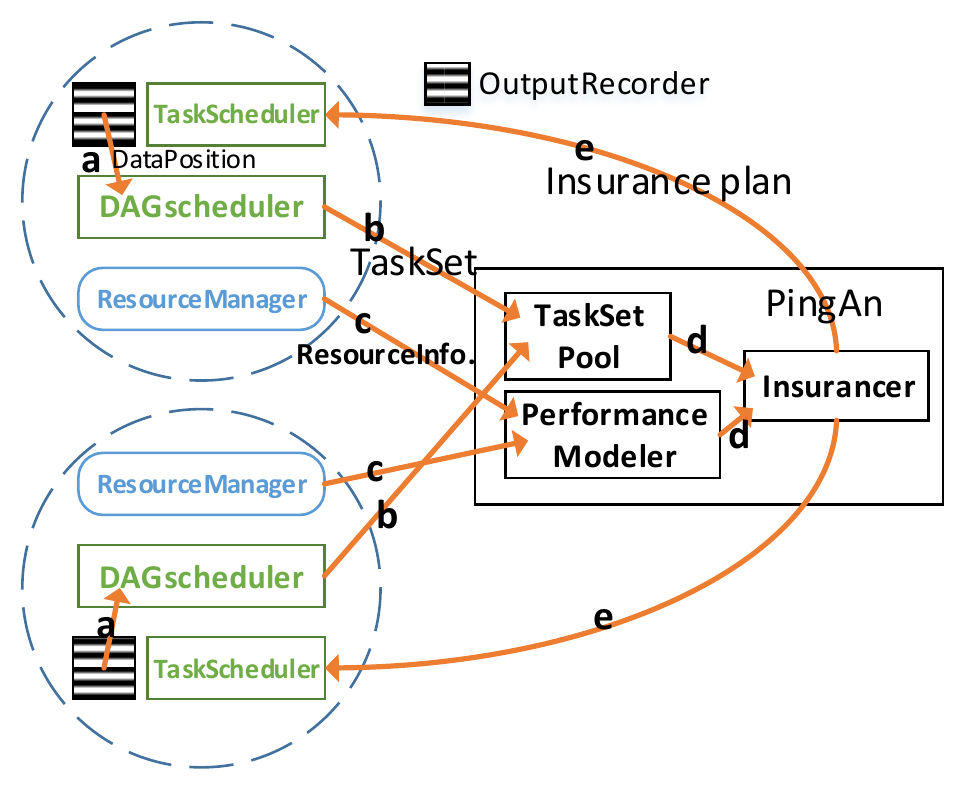}}
 \subfigure[The collection of execution information]{
    \label{fig:p2}
    \includegraphics[width=0.8\linewidth]{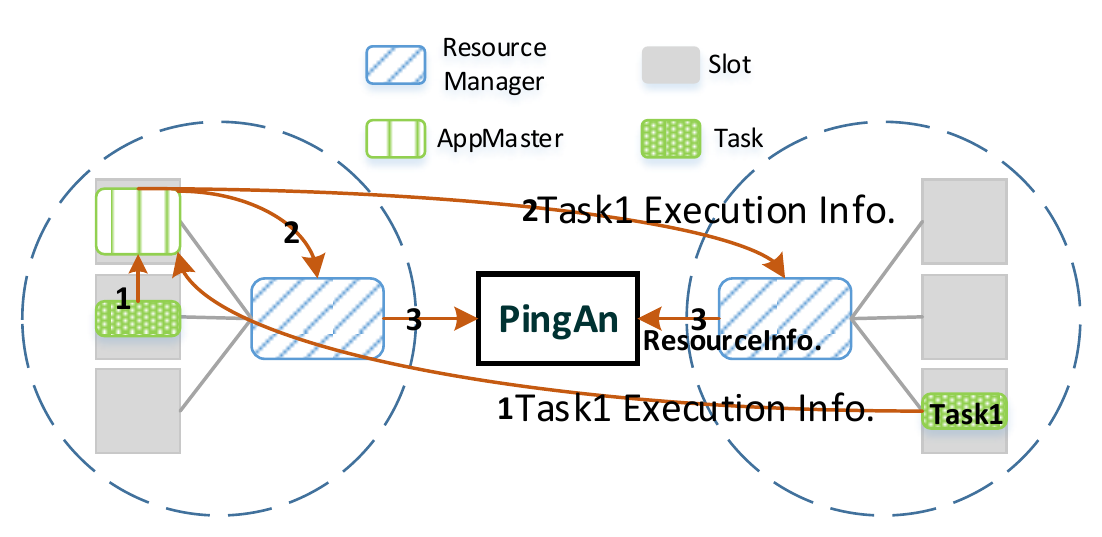}}
  \caption{The workflow of PingAn in Spark on Yarn.}
  \label{fig:PingAn_architecture}
\end{figure}
The collection of execution information is depicted in Figure \ref{fig:p2}.
As we see, {\tt AppMaster} has two tasks running in two clusters (dash line).
\begin{enumerate}[1)]
\item
After a task finishing its work, it reports its output location and execution information to {\tt AM}.
The execution information contains the data processing speed and the inter-cluster transfer speed (with a specification of the two end clusters) obtained by the task.
\item
{\tt AM} sends the execution information to {\tt RM} of the cluster running the task.
\item
{\tt RM} periodically sends the collected information to {\tt PM} in PingAn.
{\tt PM} builds resource performance models and serves for {\tt Insurancer}.
The modeling is described in Section \ref{ch:Quantification}.
\end{enumerate}
\subsection{Quantification of Cluster Selection's Impact on Execution}\label{ch:Quantification}
In this subsection, we first use the execution information to model the resource performance and quantify the impact of multi-copy execution on the task's efficiency.
Then, we quantify the multi-copy execution's impact on task's reliability.


For quantifying the impact on efficiency, first, we tally the data processing speed of recent tasks in a cluster $m$ and obtain a distribution $f_{m}^{P}(v)$ of the speed value $V^{P}_m$ to reflect the unstable computing capacity of cluster $m$.
Specifically, a newly launched task has $f_{m}^{P}(V_1)$ probability to run with $V_1$ processing speed in cluster $m$.

Notice that we use data processing speed instead of cpu processing speed because that the latter is tiring to monitor and it is impossible to be used to estimate task's time due to unit difference.
In order to eliminate the data processing speed bias caused by task type, we meticulously model such a distribution for each RDD operation which composes the Spark job and ensure a task's data processing speed distribution according to its operation.

Second, in the same way, we use $f_{m_{1},m_{2}}^{T}(v)$ to denote the distribution of data transfer bandwidth $V_{m_{1},m_{2}}^{T}$ from cluster $m_{1}$ to $m_{2}$ to reflect the unstable cluster-pair's transfer capacity.
The bandwidth of a transfer is captured at the download end.
A task $\xi_{l}^{i}$ may need multiple transfer and denoted by $\mathcal{I}_{l}^{i}$ its input location set.
When the task launches on cluster $m$, its average transfer bandwidth is given as follows
\begin{sequation*}
V^{T}_m=\frac{1}{\left | \mathcal{I}_{l}^{i} \right |}\sum_{{m}' \in \mathcal{I}_{l}^{i}}V_{m,{m}'}^{T}
\end{sequation*}
Define $f_{m}^{T}(v)$ as the distribution of $V^{T}_m$.

The execution rate of a task $\xi_{l}^{i}$ is denoted as $r_{l}^{i}(x)$ where $x$ is the number of the task's copies.
The execution rate of a task's copy hinges on the bottleneck of data transfer and data processing.
Thus, we have $V_{m}^{l,i}=\min(V_{m}^{P},V_{m}^{T})$ to indicate the execution rate of a copy in cluster $m$.
Let $f_{m}^{l,i}(v)$ be the corresponding distribution of $V_{m}^{l,i}$.
Both of the distribution $f_{m}^{T}(v)$ and $f_{m}^{l,i}(v)$ can be easily derived via the composition computation of multiple discrete random variables.
If a task has just one copy in cluster $m_1$, we have $r_{l}^{i}(1)=\mathbb{E}[V_{m_1}^{l,i}]$.
If the task has another extra copy in cluster $m_2$, then $r_{l}^{i}(2)=\mathbb{E}[\max(V_{m_1}^{l,i},V_{m_2}^{l,i})]$ and so on.

For quantifying the impact on reliability, first, let $p_{m}$ be the probability of encountering cluster-level unreachable trouble in cluster $m$, which is obtained via the statistic.
We assume time to be slotted and the cluster-level failure are independent over time.
Hence, we can assume the cluster-level failure to follow a binomial distribution.
To recap, a copy in cluster $m$ has $1-p_{m}$ probability to exempt from cluster-level trouble at each time slot.
${pro}_{l}^{i}$ denotes the probability of the task $\xi_{l}^{i}$ encountering no cluster-level trouble during its execution and is given by
\begin{sequation*}
{pro}_{l}^{i}=\binom{e_{m}^{l,i}}{0} {p_{m}^{0}(1-p_{m})^{(e_{m}^{l,i}-0)}}=(1-p_{m})^{e_{m}^{l,i}}
\end{sequation*}
where $e_{m}^{l,i}=\frac{datasize}{r_l^i(1)}$ is the execution time of task $\xi_l^i$ in cluster $m$.
If the task has one extra copy in the same cluster, ${pro}_{l}^{i}=(1-p_{m})^{\frac{datasize}{r_l^i(2)}}$ where the probability of task encountering troubles at each time slot is invariant because that once the cluster-level trouble happens, both two copies fail.
If the task has one extra copy in the other cluster $m_a$, the trouble encountering probability of the task at a time slot decreases to $p_{m}*p_{m_{a}}$, thus we have ${pro}_{l}^{i}=(1-p_{m}p_{m_{a}})^{\frac{datasize}{r_l^i(2)^'}}$.
More copies are as the same analogy.
\subsection{Formulation}
Consider a system consisting of $M$ clusters denoted by the set $\mathcal{K}$.
The clusters' topology conforms to a heavy-tailed distributions, which means that each large-scale data center are linked by multiple small edges and multiple data centers are interconnected ~\cite{Tong2016A}.
Some neighboring edges are also connective.
$M_k$ denotes the number of computing slots in the $k\text{-th}$ cluster.
The egress and ingress bandwidth restriction for the $k\text{-th}$ cluster is referred to as ${Ing}_k$ and ${Eg}_k$ respectively.

Assume a set of jobs $\mathcal{J}=\{J_1,J_2,\dots\}$ arriving over time. Job $J_i \in \mathcal{J}$ arrives at time $a_i$ and consists of $n$ tasks $\mathcal{L}_i=\{\xi_1,\xi_2,\dots,\xi_n\}$.
The flowtime of job $J_i$ is $f_i-a_i$ where $f_i$ is the job's completion time.
Our insurance problem aims to minimize the sum of job flowtimes's expectation.
The formulation is outlined below.

In formulation, $x_{l,k}^{i}=c$ indicates insuring $c$ copy for task $\xi_{l}^{i}$ in cluster $k$.
As an effective insurance, constraint in Eq. (\ref{c1}) ensures each task at least one copy to complete its work (Here, task with just one copy means that task executes without speculation).
Eq. (\ref{c:c4}) states that a task in a job can only be scheduled after the job's arrival time.
Eq. (\ref{c:c5})-Ep. (\ref{c:c8}) ensures the execution of tasks satisfying the precedence order in a job.
Let $\leqslant^{i}$ represents the partial order in job $J_i$.
The start time of tasks ($st(.)$) should obey the partial order $\leqslant^{i}$, which means that for each pair of ordered tasks, $(\xi_{u},\xi_{v})\in \leqslant^{i}$, it satisfies $st(\xi_{v}) \geq st(\xi_{u})+e_{\xi_{u}}$ where $e_{\xi_{u}}$ is the execution time of task $\xi_{u}$.
In addition, $D_l^i$ in Eq. (\ref{c:c6}) denotes the task's datasize and $f_{l}^{i}$ in Eq.(\ref{c:c7}) denotes the completion time of task $\xi_{l}^{i}$.
Constraints in Eq. (\ref{c2}), Eq. (\ref{c3}) and Eq. (\ref{c4}) forestall the terrible contention for cluster's computing slots and gate bandwidths at any time slot.
Eq. (\ref{c:c12}) is due to that the completion time of a job $f_i$ depends on its last task's completion time.

\textbf{The Difficulty of Insurance Problem:}
Without considering the speedup via task cloning and limiting the gate bandwidth of clusters, our problem can be simplified to the scheduling problem in ~\cite{zheng2013new} which has proven to be NP-hard.
Therefore, our problem is naturally NP-hard and we devote to solve the problem with an approximation bound.
\begin{eqnarray}
	\label{objective_1}\mathop{\text{min}}
	\sum_{i}\mathbb{E}[f_{i}-a_{i}] \qquad \qquad \qquad \qquad \\
	\text{s.t} \qquad \qquad \qquad
     x_{l,k}^i\in\mathbb{N}  \qquad \forall i;\forall l, k \qquad \qquad \qquad \\
	 \label{c1}\sum\limits_{k\in \mathcal{K}} x_{l,k}^{i} \ge 1  \quad \forall i;\forall l\in\mathcal{L}_{i} \qquad \qquad \qquad \\
     \label{c:c4}st(\xi_{l}^{i})\geqslant a_{i}    \quad \forall i;\forall l\in \mathcal{L}_{i} \qquad \qquad \qquad \\
     \label{c:c5}x_{l}^{j} = \sum\limits_{k\in \mathcal{K}} x_{l,k}^{i} \quad \forall i;\forall l\in \mathcal{L}_{i} \qquad \qquad \qquad \\
     \label{c:c6}\mathbb{E}[e_{l}^{i}]={D_{l}^{i}}/{r_{l}^{j}(x_{l}^{j})}\quad \forall i;\forall l\in\mathcal{L}_{i} \qquad \qquad \\
     \label{c:c7}f_{l}^{i}=st(\xi_{l}^{i})+e_{l}^{i} \quad \forall i;\forall j\in\mathcal{L}_{i} \qquad \qquad \\
     \label{c:c8}st(\xi_{u}^{i})\geq st(\xi_{v}^{i})+e_{v}^{i} \quad \forall i; \forall(v,u)\in\leqslant^{i} \qquad \qquad \\
	 \label{c2}\sum\limits_{^{i\in\mathcal{J};l\in \mathcal{L}_{i};}_{st(\xi_{l}^{i})\geq t;f_{l}^{i}\leq t}} x_{l,k}^{i} \le M_k  \quad \forall k\in \mathcal{K};\forall t \qquad \qquad \\
	 \label{c3}\sum\limits_{^{i\in\mathcal{J};l\in \mathcal{L}_{i};}_{st(\xi_{l}^{i})\geq t;f_{l}^{i}\leq t}}\sum\limits_{d=1}^{I_{l}^{i}} x_{l,k}^{i}\cdot \mathbb{E}[F_{d,k}^{T}]\le {Ing}_k \quad \forall  k\in \mathcal{K};\forall t\\
	 \label{c4}\sum\limits_{^{i\in\mathcal{J};l\in\mathcal{L}_{i};}_{st(\xi_{l}^{i})\geq t;f_{l}^{i}\leq t}}\sum\limits_{^{\quad k\in M;}_{d\in I_{l}^{i}, d = w}} x_{l,k}^{i}\cdot \mathbb{E}[F_{d,k}^{T}] \le {Eg}_w \quad \forall  w\in \mathcal{K};\forall t \\
	 \label{c:c12}f_{i}=\max_{l\in \mathcal{L}_{i}}f_{l}^{i} \quad \forall i \qquad \qquad \qquad
\end{eqnarray}
\section{PingAn Insurance}\label{PingAn}
\subsection{Algorithm Design}
Motivated by the work in ~\cite{Xu2015Task} which speed up the online job via task cloning in a single cluster, we extend its idea that the jobs with the least remaining workload shares the machines of a cluster and design our PingAn insurance algorithm.
PingAn works at the beginning of each time slot.
Out of practicality, we only use the job knowledge available at the current scheduled stage.
The effective workload of a job can be characterized by the unprocessed data size of its current stage.
The jobs with higher priority has less unprocessed data size than jobs with lower priority.

Let $N(t)$ be the number of alive jobs at time $t$ and $\varepsilon$ be a value in $(0,1)$.
The first $\varepsilon N(t)$ jobs with the least unprocessed data size fairly shares the geo-distributed slots, which means that each prior job is admitted to obtain at most $h_{i}(t)=\left \lceil {\frac{\sum\limits_{k\in \mathcal{K}}M_{k}}{\varepsilon N(t)}}\right \rceil$ slots and the other jobs with lower priority can obtain nothing.
After deciding the promissory slots for each job, PingAn insures the cluster-specified slots to each task after multiple rounds.
Notice that, in any round of insurance, the total slots number insured for a job is limited to its promissory slot number $h_{i}$.

\begin{algorithm}[!tbp]
  \caption{PingAn Insurance Algorithm}\label{alg:PingAn}
  \KwIn{$\mathcal{J}(t)$, the set of alive jobs at current time slot $t$;
	
	$M_{\mathcal{K}}(t)$, the available slots at time $t$;

	$V^A_{\mathcal{K}}(t), V^T_{(\mathcal{K},\mathcal{K})}(t), p_{\mathcal{K}}(t)$, the resource condition at time $t$;

	$Ing_{\mathcal{K}}, Eg_{\mathcal{K}}$, the gate bandwidth limit of clusters;}
  \KwOut{An insurance plan}
  Sort the jobs in $\mathcal{J}(t)$ according to the ascend order of unprocessed datasize\;
  \For{each $J_i\in \mathcal{J}(t)$}
  {
    Compute $g_{i}(t)$, the slot number promised to job $J_i$\;
    Count $\theta_i(t)$, the number of slots running $J_i$'s tasks\;
  }
  $N_{slot}=0$, the assigned slot number in a round\;
  \For{each $J_i\in \mathcal{J}(t)$ and $g_i(t)-\theta_i(t)>0$}
  {
    Extract waiting tasks to $\mathcal{L}_i^{0}$\;
    \For{each $\xi_{l}\in \mathcal{L}_i^{0}$ and $g_i(t)-\theta_i(t)>0$}
    {
        Try to do efficient-first insurance for $\xi_{l}$\;
        \If{insuring succeeds}
        {
         $\theta_i(t)$++, $N_{slot}^{i}$++\;
        }
    }
  }
  \If{$N_{slot}$ == 0}
  {
    return\;
  }
  $N_{slot}=0$\;
  \For{each $J_i\in \mathcal{J}(t)$ and $g_i(t)-\theta_i(t)>0$}
  {
    Extract tasks assigned a slot to $\mathcal{L}_i^{1}$\;
    Compute $pro_l^i$ for each $\xi_l\in \mathcal{L}_i^{1}$\;
    Sort tasks in $\mathcal{L}_i^{1}$ in the ascend order of $pro_l^i$\;
    \For{each $\xi_{l}\in \mathcal{L}_i^{1}$ and $g_i(t)-\theta_i(t)>0$}
    {
        Try to do reliability-aware insurance for $\xi_{l}$\;
        \If{insuring succeeds}
        {
         $\theta_i(t)$++, $N_{slot}^{i}$++\;
        }
    }
  }
  \If{$N_{slot}$ == 0}
  {
    return\;
  }
  \While{true}
  {
      $N_{slot}=0$\;
      \For{each $J_i\in \mathcal{J}(t)$ and $g_i(t)-\theta_i(t)>0$}
      {
        Extract tasks copied in the last round to $\mathcal{L}_i^{\geqslant2}$\;
        \For{each $\xi_{l}\in \mathcal{L}_i^{\geqslant2}$ and $g_i(t)-\theta_i(t)>0$}
        {
            Try to do resource-saving insurance for $\xi_{l}$\;
            \If{insuring succeeds}
            {
             $\theta_i(t)$++, $N_{slot}^{i}$++\;
            }
        }
      }
      \If{$N_{slot}$ == 0}
      {
        return\;
      }
  }
\end{algorithm}

In the first round, PingAn only insures at most one slot for each task in order of job priority according to an efficiency-first principle.
When its turn arrives, a task can obtain a slot and run with currently the best execution rate $\mathbb{E}[r_{l}^{i}(1)]$, as long as the related gate bandwidth restrictions are satisfied and the execution rate $\mathbb{E}[r_{l}^{i}(1)]$ is not worse than $1/(\varepsilon+1)$ fraction of the global optimal rate $\mathbb{E}^O[r_{l}^{i}(1)]$ obtained by the task when only it executes in the system.
If the bandwidth are not enough or the current best slot's rate is too worse, the task waits for the next insurance.

In the second round, PingAn utilizes the current idle slots to improve the reliability of each job in priority order following reliability-aware principle.
Inside each job, PingAn prefers to insure an extra copy for the tasks with the worse trouble-exemption probability ${pro}_{l}^{i}$.
After meeting the bandwidth restrictions and the lower limit of execution rate, the slot is selected from the cluster where the copy execution can improve the task's ${pro}_{l}^{i}$ to the greatest extent.

In the third or the later round, beside following the efficient-first principle, PingAn starts to consider the opportunity cost of a slot being an extra copy since the slot can be saved to complete many more tasks in the next insurances.
Given that an insured task in the third round already has a copy with the best efficiency and an extra copy improving the execution reliability, a slot used to run the third copy of the task plays a relatively less role on performance improvement than using the slot to run the first or second copy of the other tasks.
Thus, in the third and the later round, PingAn conservatively insures a copy for a task only if it saves both time and resources consumed, which is referred to as resource-saving copy.
To be specific, supposing to decide whether to schedule the $c\text{-th}$ copy of a task in cluster $k$ ($c\geq 2$), PingAn calculates the execution rate $\mathbb{E}[r_{l}^{i}(c)]$ and the corresponding execution time $\mathbb{E}^{c}[e_{l}^{i}]$ of the task if the extra copy performs.
Only if the $\mathbb{E}^{c-1}[e_{l}^{i}] > \frac{c+1}{c}\mathbb{E}^{c}[e_{l}^{i}] $, the extra copy is permitted to insuring for the task. Algorithm \ref{alg:PingAn} summarizes in detail how PingAn insures for the jobs.

As applied in PingAn, the efficient-first principle means that the efficiency should be satisfied priori to the reliability for a task execution, which is motivated by the factor that resource performance fluctuates frequently but the cluster-level trouble is occasional.
Thus, when PingAn insures the first slot to the task, aiming at the efficiency can drastically and directly reduce the execution time.

The efficiency-first principle in the first round also enforces that the efficiency of a job with promissory slots should be satisfied priori to the reliability of a job with higher priority.
Recalling that our objective is to minimizing the sum of job flowtimes, a slot used for job's reliability generally contributes less to the objective than a slot used for job's efficiency since the former has less chance to save the execution time from occasional cluster-level troubles.
Therefore, PingAn insures only essential copy for all qualified jobs' tasks in the first round and insures extra copies in the later rounds, which is referred to as Efficient-First Allocation (EFA) among jobs.
The other alternative is to insure both essential and extra copies for each job in priority order, which is referred to as Job Greedy Allocation (JGA) among jobs.
We compare the practical performance of two candidates in Section \ref{ch:principleExp} to verify the correctness of efficient-first principle.

The reliability-aware part is to emphasize that in despite of the occasionality, a cluster-level trouble can harm a large scale of job performances. Thus, the reliability is indispensable towards performance improvement.
We verify the practical effect of the efficient-first reliability-aware principle in Section \ref{ch:principleExp}.
\subsection{Analysis of PingAn}
In this section, we assumes resource augmentation ~\cite{Berman1998Speed} for PingAn and derive the approximation bound via the method of potential function analysis, .
In the analysis, time is continuous and we do not consider the cluster-level unreachable troubles in our approximate bound analysis because that the impact of such failures on job flowtime is hard to be measured.
Before the details of potential function analysis, we first prove the following Proposition \ref{proposition1} and deduce the final approximation bound with the help of it.

\begin{proposition}\label{proposition1}
For any integer $b\geq a > 0$, we have $\frac{r_{l}^{i}(a)}{a}\geq \frac{r_{l}^{i}(b)}{b}$ under PingAn algorithm.
\end{proposition}
\begin{proof}
See Appendix \ref{ch:appendix1}.
\end{proof}
Depending on Proposition \ref{proposition1}, we prove the following Theorem \ref{the:third}.
Under a resource augmentation assumption, the resource speed in PingAn is $1+\varepsilon$ times faster than the one in the optimal adversary algorithm.
Theorem \ref{the:third} states that the sum of job flowtimes in PingAn is within $o(\frac{1}{\varepsilon^2+\varepsilon})$ factor of the optimal algorithm with a resource augmentation.
\begin{theorem}\label{the:third}
PingAn is $(1+\varepsilon)\text{-speed}\ o(\frac{1}{\varepsilon^2+\varepsilon})\text{-competitive}$ approximated algorithm for the sum of the expectation of job flowtimes when $0<\varepsilon<1$.
\end{theorem}
\begin{proof}
See Appendix \ref{ch:appendix2}.
\end{proof}
\section{Implementation on Real System}\label{ch:systemImple}
We develop the PingAn in Spark on Yarn, a general geo-distributed data analysis system, and run a series of typical workloads to consolidate the practicality and acceleration ability of PingAn.

\textbf{Testbed:}
Our experiments are deployed on 10 VMs running a 64-bit Ubuntu 16.04.
Four of them have 8 CPU cores and 20GB memory, the others have 4 CPU cores and 10GB memory.
We regard the 10 VMs as ten different edge clusters and the number of containers concurrently running on the VM corresponds to the computing slots number in the edges.
We run two {\tt ResourceManager}s in charge of 5 VMs respectively.
We use the Wondershaper to limit the egress and ingress bandwidth of each VM.
We intentionally run benchmarks in each VM to consume its spare resources to different extent (Ubench for CPU and memory, Bonnie for disk I/O and Iperf for external bandwidth) in order to cause performance difference via resource contention.
In addition, a script file is running for executing shutdown command in VM according to the preset probability to imitate the cluster-level errors.
The adjustable parameter $\varepsilon$ in PingAn is set to be 0.6.

\textbf{Applications:}
The workload includes 88 jobs such as WordCount, Iterative machine learning and PageRank. The variation in input sizes is based on real workloads from Yahoo! and Facebook ~\cite{vavilapalli2013apache} with a reduced scale as shown in Table \ref{tab:RealJobWorkload}.
We randomly distribute the input across the 10 VMs.
The job submission time follows an exponential distribution.
The average workload intensity is 3 jobs per 5 min.

\textbf{Baseline:}
We compare PingAn with the Spark with delay scheduling for tasks and fair scheduling for jobs and the speculative Spark when Spark's default speculation mechanism works.

\textbf{Metric:}
We focus on the average flowtime of jobs and the cumulative distribution function (CDF) of job flowtimes.

\begin{table}[!htbp]
\setlength{\abovecaptionskip}{2pt}
\setlength{\belowcaptionskip}{-5pt}
\caption{Workload Constitution}
\centering
\begin{tabular}{l|c|c|c}
\hline
{JobType}&WordCount&Iterative ML&PageRank \\
\hline
Small$(46\%)$&100-200MB&130-300MB&150-400MB\\
\hline
Medium$(40\%)$&0.7-1.5GB&1.3-1.8GB&1-2GB\\
\hline
Large$(14\%)$&3-5GB&2.5-4GB&3.5-6GB\\
\hline
\end{tabular}
\label{tab:RealJobWorkload}
\end{table}
\vspace{-0.2cm}
\begin{figure}[!htbp]
\setlength{\abovecaptionskip}{0pt}
\setlength{\belowcaptionskip}{-10pt}
\centering
\includegraphics[width=5.5cm,height=3.5cm]{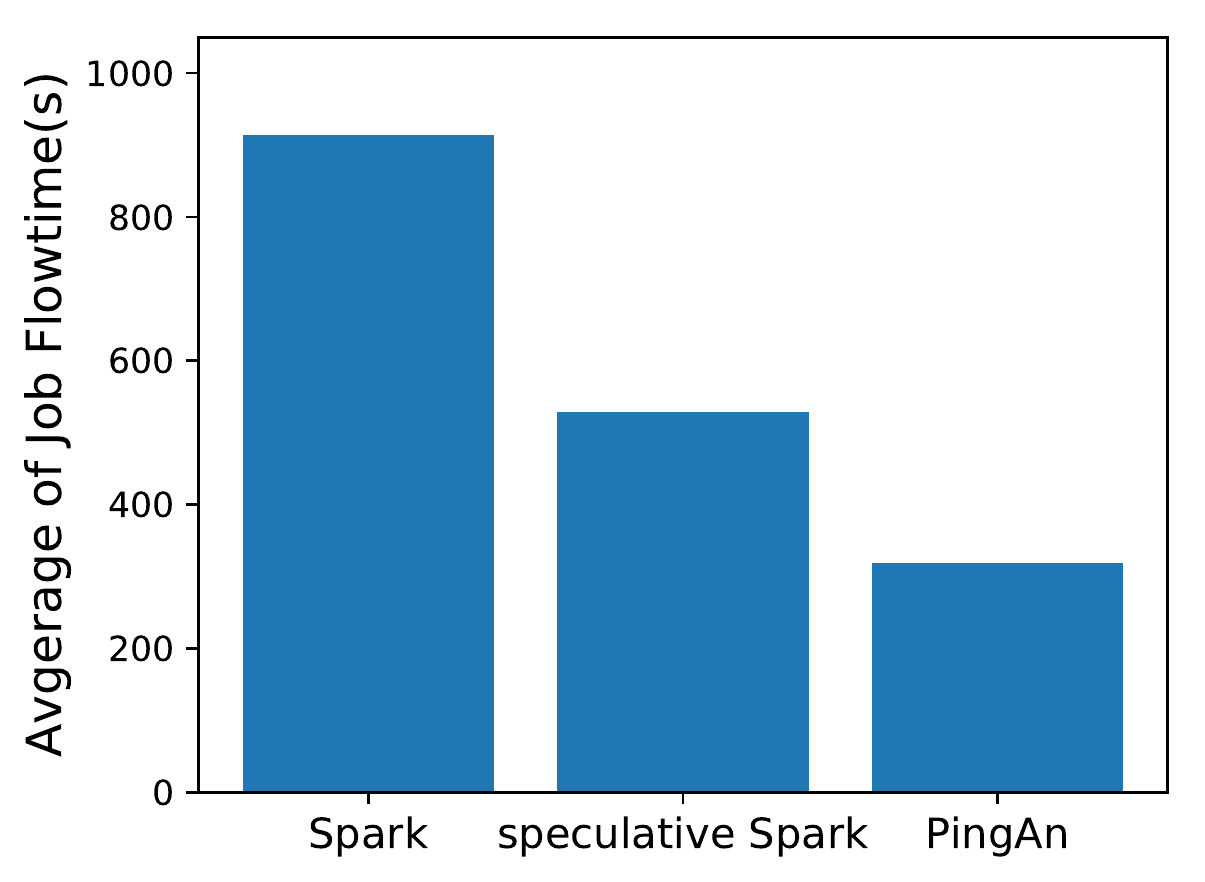}
\caption{The average job flowtime comparison under PingAn, Spark and speculative Spark execution.}
\label{fig:Spark_avg}
\end{figure}
\vspace{-0.2cm}
\begin{figure}[!htbp]
\setlength{\abovecaptionskip}{0pt}
\setlength{\belowcaptionskip}{-10pt}
  \centering
  \subfigure[The flowtime CDF of jobs with $<$500s flowtime.]{
    \label{fig:s1}
    \includegraphics[width=0.48\linewidth]{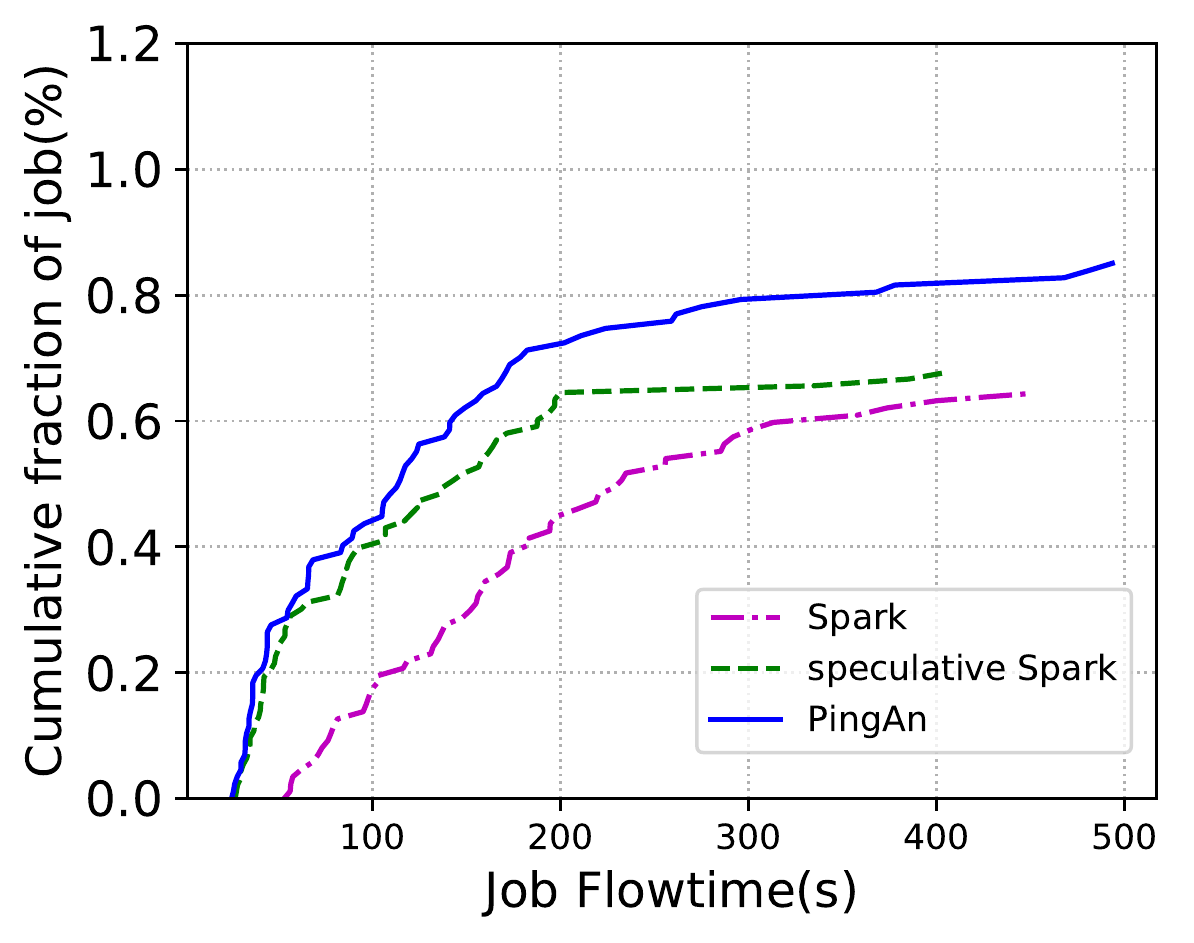}}
 \subfigure[The flowtime CDF of jobs with $>$300s flowtime.]{
    \label{fig:s2}
    \includegraphics[width=0.48\linewidth]{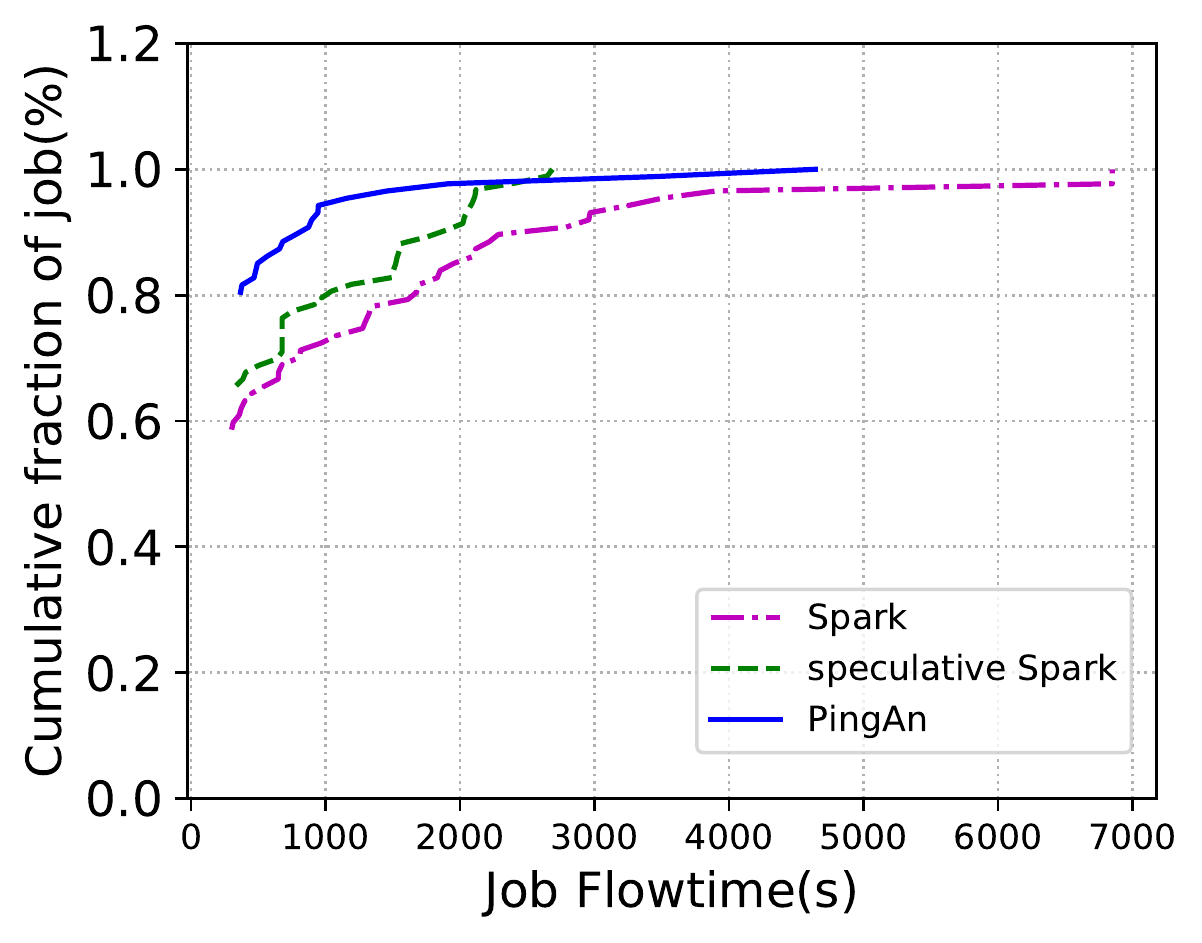}}
  \caption{The CDF of job flowtimes under PingAn, Spark and speculative Spark execution.}
  \label{fig:Spark_cdf}
\end{figure}

Figure \ref{fig:Spark_avg} shows that PingAn reduces the average job flowtime by 39.6$\%$ comparing to the default speculation mechanism in Spark.
As exhibited concretely in Figure \ref{fig:Spark_cdf}, PingAn efficiently reduces the job flowtimes via coordinating resource contention among jobs and insuring proper copies for tasks.
Figure \ref{fig:s1} depicts the CDF of flowtimes for jobs whose flowtime is between 0 and 500 seconds under three algorithm.
It indicates that 72.4$\%$ jobs in PingAn finishes within 200 seconds while the proportion in speculative Spark and Spark are 65.6$\%$ and 45.9$\%$ respectively.

Figure \ref{fig:s2} depicts the CDF of flowtime for jobs whose flowtime is larger than 300 seconds.
It shows that the detect-based speculation mechanism in Spark inhibits the overlong tasks.
PingAn arranges copies at the execution start to avoid straggler, thus, it is helpless in face of slow tasks, while it saves the system's cost of remote monitoring.
\section{Performance Evaluation}
\begin{table*}[!htbp]
\centering
\caption{Simulation Experiments settings}
\centerline{
\begin{tabular}{l|c|c|c|c|c|c|c|c}
\hline
\multirow{2}*{ClusterType}& \multirow{2}*{Proportion}& \multirow{2}*{VM Number}& \multirow{2}*{$^{\text{Gate Bandwidth}}_{\ \;\text{Limit Ratio}}$}& \multicolumn{2}{c|}{VM Power}& \multicolumn{2}{c|}{WAN Bandwidth}& \multirow{2}*{$^{\text{Unreachability}}_{\ \text{Probability}}$}\\
\cline{5-8}
{}&{}&{}&{}&Mean(mips)&$^{\text{Relative Standard}}_{\text{Deviation (RSD)}}$&Mean(kb/s)&RSD\\
\hline
Large-scale&5$\%$&500-1500&55$\%$-75$\%$&174-355&0.25-0.6&\multirow{3}*{64-256}&\multirow{3}*{0.2-0.5}&0.002-0.011\\
\cline{1-6}
\cline{9-9}
Medium-scale&20$\%$&50-500&65$\%$-85$\%$&128-241&0.55-0.85&&&0.02-0.2\\
\cline{1-6}
\cline{9-9}
Small-scale&75$\%$&10-50&75$\%$-95$\%$&68-179&0.35-0.75&&&0.05-0.5\\
\hline
\end{tabular}
}
\label{tab:SimulationSetUp}
\end{table*}
In simulations, we expand experiments' scale to verify the acceleration effect of PingAn on different conditions.
\subsection{Methodology}
\textbf{Simulation Setup:}
We modify CloudSim to support our simulation experiments.

The clusters launched in simulation have large, median and small three kinds of scale.
We use the BRITE Topo generator to create 100 clusters with a heavy-tailed distribution around the world.
We sort 100 clusters in the decreasing order of their degrees and let the first $5\%$ clusters be the large-scale cluster, the following $20\%$ be the medium one and the rest be the small one.

Table \ref{tab:SimulationSetUp} shows the various kinds of parameters' range setting in different scale cluster.
Some parameters' range (VM Power and WAN Bandwidth) is based on the real performance analysis experiments on Amazon EC2 or other public clouds ~\cite{dejun2010ec2,schad2010runtime,Zaharia2008Improving}, and some parameters' range is set to be wide for excavating the ability of our algorithm.
Specially, the Gate Bandwidth Limit Ratio in the fourth column of Table \ref{tab:SimulationSetUp} indicates the ratio of the egress/ingress bandwidth to the sum of the VMs' external bandwidth of a cluster.
We assumes the VM power and inter-cluster bandwidth to follow a normal distribution as observed in ~\cite{schad2010runtime}.

\textbf{Workloads:}
We construct an synthetic workloads containing 2000 Montage workflows.
Montage workflow assembles high-resolution mosaics of region of the sky from raw input data, which consist of the tasks with high demand of both data transfer and computing.
The job size distribution refers to the traces in Facebook's production Hadoop cluster ~\cite{Ananthanarayanan2010Reining,Ananthanarayanan2013Effective,Ananthanarayanan2014GRASS,ananthanarayanan2012pacman} that $89\%$, $8\%$ and $3\%$ of jobs are with small (1-150), medium (151-500) and large ($>$500) task numbers respectively.
We randomly disperse the raw input data of each workflow to the edges as well as some medium-scale clusters.
The workflow inter-arrival times are derived from a Poisson distribution.
We adjust system load condition via the Poisson parameters $\lambda$ from 0.02 to 0.15.

\textbf{Baseline:}
We compare PingAn with four baseline algorithms.
\begin{enumerate}[(1)]
\item
Flutter. Flutter is a geo-distributed scheduler to optimize stage completion time.
\item
Iridium. Iridium optimizes data and task placement to reduce the WAN transfer during the job execution.
\item
Flutter+Mantri.
Mantri is demonstrated to be the best detection-based speculation mechanism inside cluster.
\item
Flutter+Dolly.
Dolly is a passive cloning mechanism and performs better than Mantri under the Facebook's trace.
\end{enumerate}

\textbf{Metric:}
We focus on the same metric in Section \ref{ch:systemImple}.
In addition, for Dolly, Mantri and PingAn, we focus on their reduction in job flowtime of the Flutter as well as the CDF of the reduction ratio.
Under each setting, we run our workloads ten times and calculate the average flowtime of the ten executions for each job as its final flowtime.
\subsection{Comparison against Baselines under Different Load}
\begin{figure}[!tp]
\vspace{-0.2cm}
\setlength{\abovecaptionskip}{-3pt}
\setlength{\belowcaptionskip}{-10pt}
    \centering
        \includegraphics[width=0.9\linewidth]{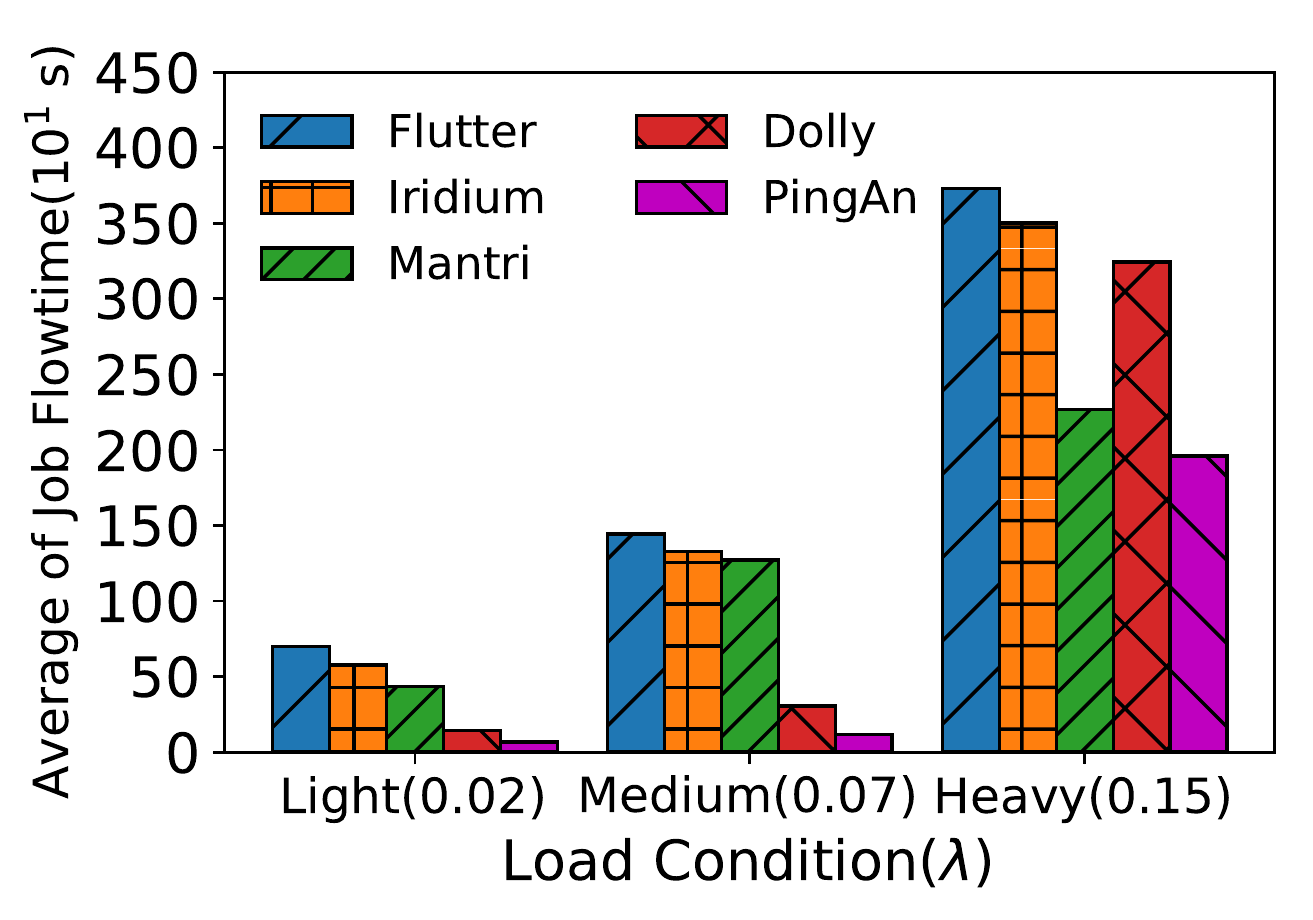}
    \caption{The performance comparison under different load condition.}
    \label{fig:performance-load}
\end{figure}
We compare the the average job flowtime of PingAn with four baselines under light, medium and heavy load respectively.
We set $\varepsilon=0.8$ for PingAn under light load, $\varepsilon=0.6$ under medium load and $\varepsilon=0.2$ under heavy load according to the $\varepsilon$ selection hint in Section \ref{ch:epsilonselection}.

Figure \ref{fig:performance-load} shows the comparison results.
Without the awareness of cluster heterogeneity, the job performance in both Flutter and Iridium keep away from the expectation.
As a whole, Dolly and Mantri have adept load case apiece and PingAn works the best on all load condition. PingAn reduces the average job flowtime by 52.9$\%$, 61.9$\%$ and 13.5$\%$ than the best baseline under light, medium and heavy load respectively.
More details are illustrated in Figure \ref{fig:flowtime_comp_under_different_load}.
\begin{figure}[!tp]
\setlength{\abovecaptionskip}{2pt}
\setlength{\belowcaptionskip}{-10pt}
  \centering
  \subfigure[CDF of flowtime under light load ($\lambda=0.02$)]{
    \label{fig:c1}
    \includegraphics[width=0.45\linewidth]{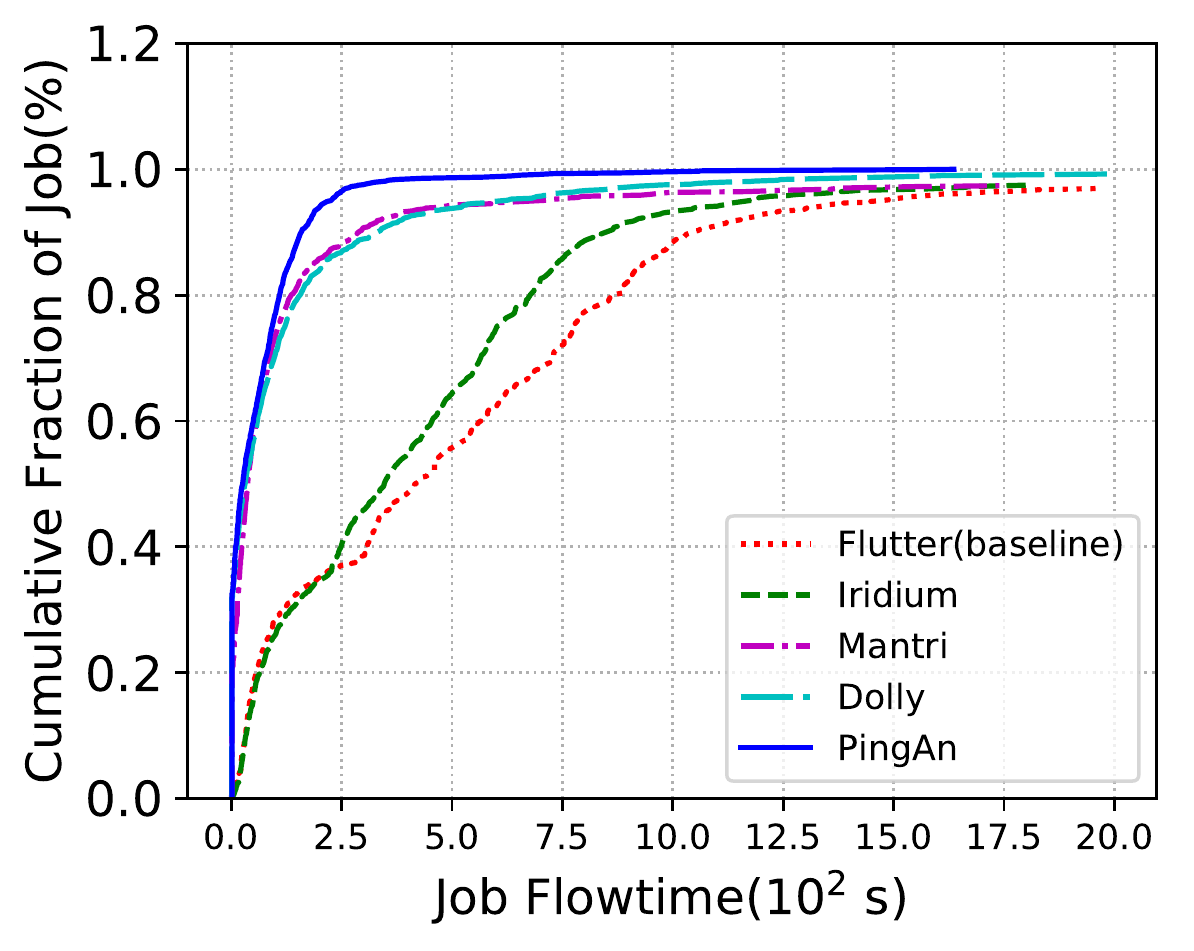}}
  \hspace{0.01\linewidth}
  \subfigure[CDF of flowtime reduction under light load ($\lambda=0.02$)]{
    \label{fig:r1}
    \includegraphics[width=0.45\linewidth]{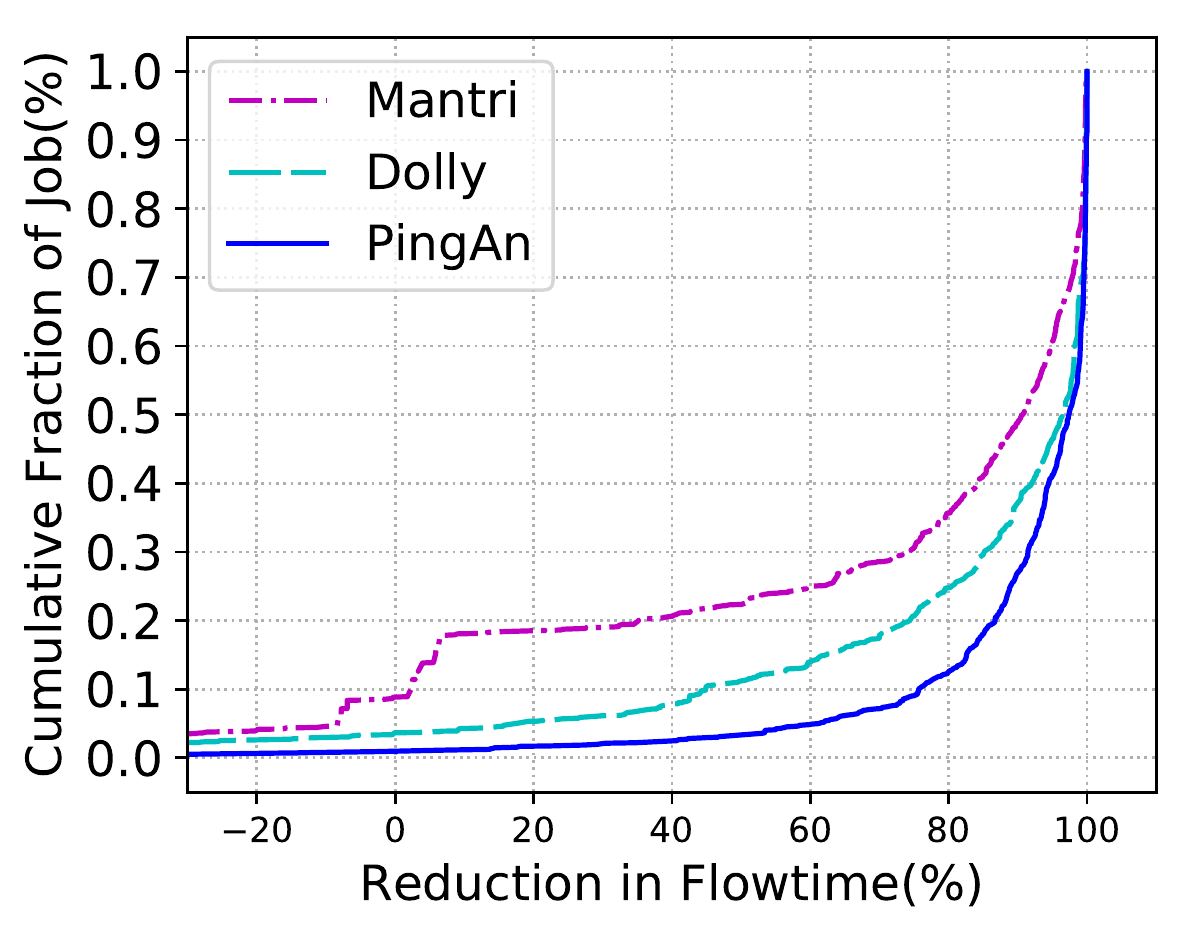}}
  \vfill
  \subfigure[CDF of flowtime under medium load ($\lambda=0.07$)]{
    \label{fig:c2}
    \includegraphics[width=0.45\linewidth]{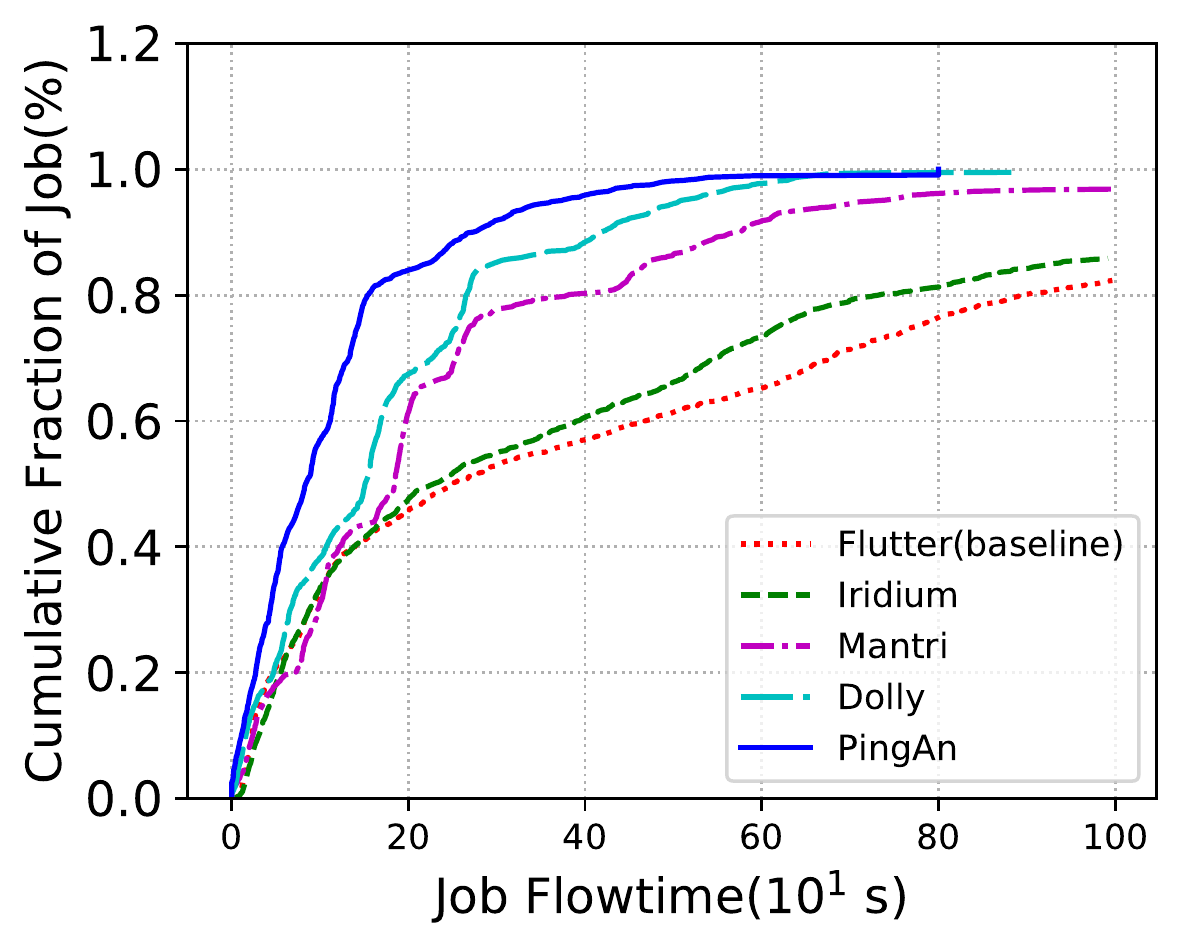}}
  \hspace{0.01\linewidth}
  \subfigure[CDF of flowtime reduction under medium load ($\lambda=0.07$)]{
    \label{fig:r2}
    \includegraphics[width=0.45\linewidth]{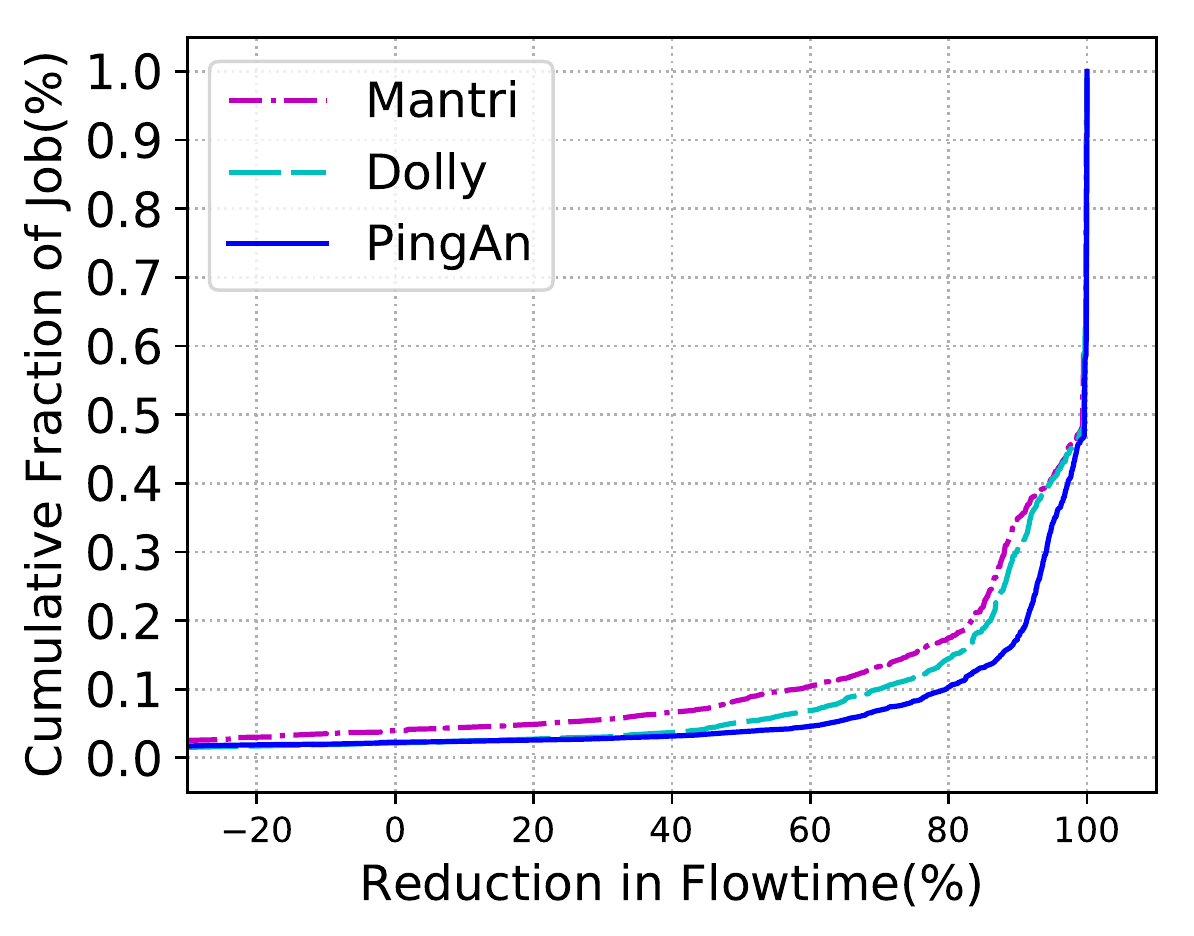}}
  \vfill
  \subfigure[CDF of flowtime under heavy load ($\lambda=0.15$)]{
    \label{fig:c3}
    \includegraphics[width=0.45\linewidth]{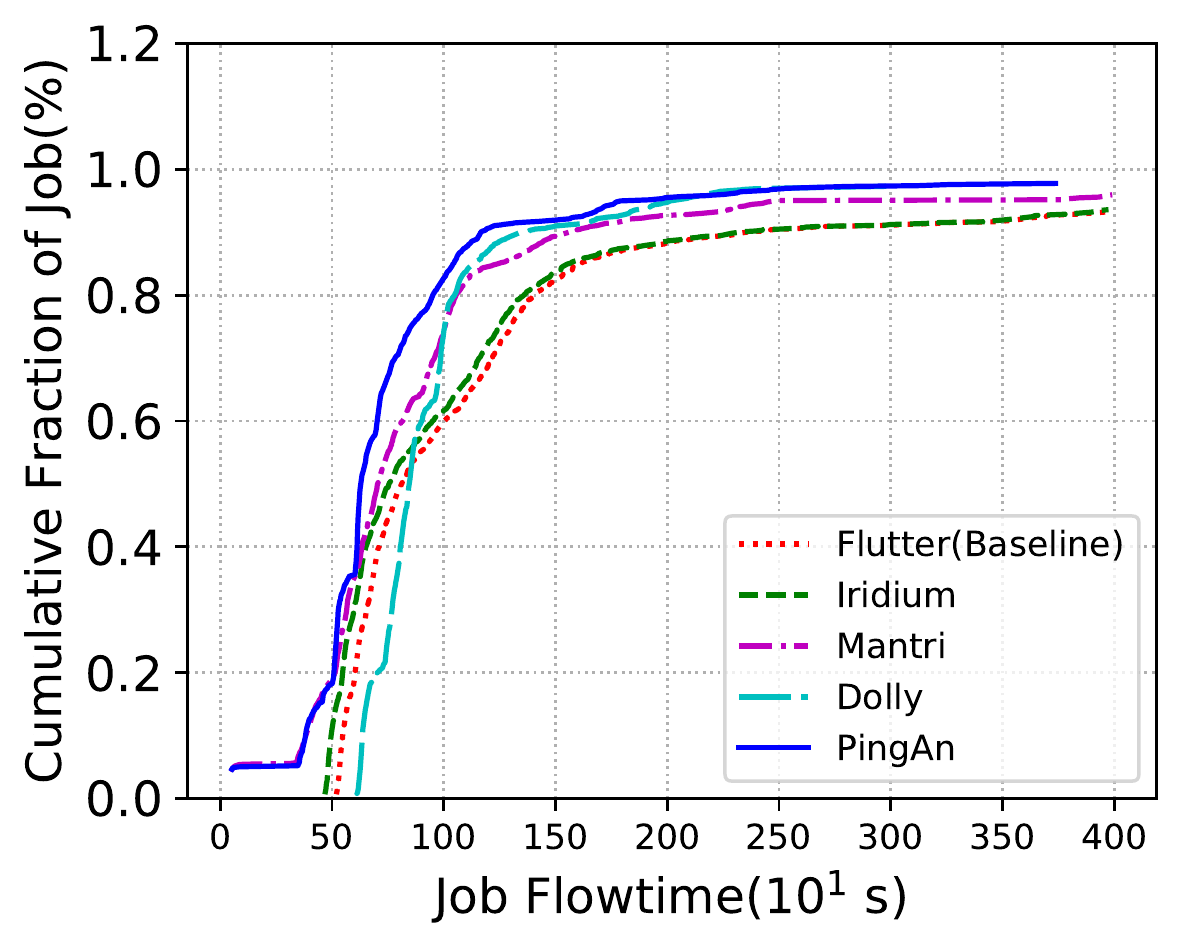}}
  \hspace{0.01\linewidth}
  \subfigure[CDF of flowtime reduction under heavy load ($\lambda=0.15$)]{
    \label{fig:r3}
    \includegraphics[width=0.45\linewidth]{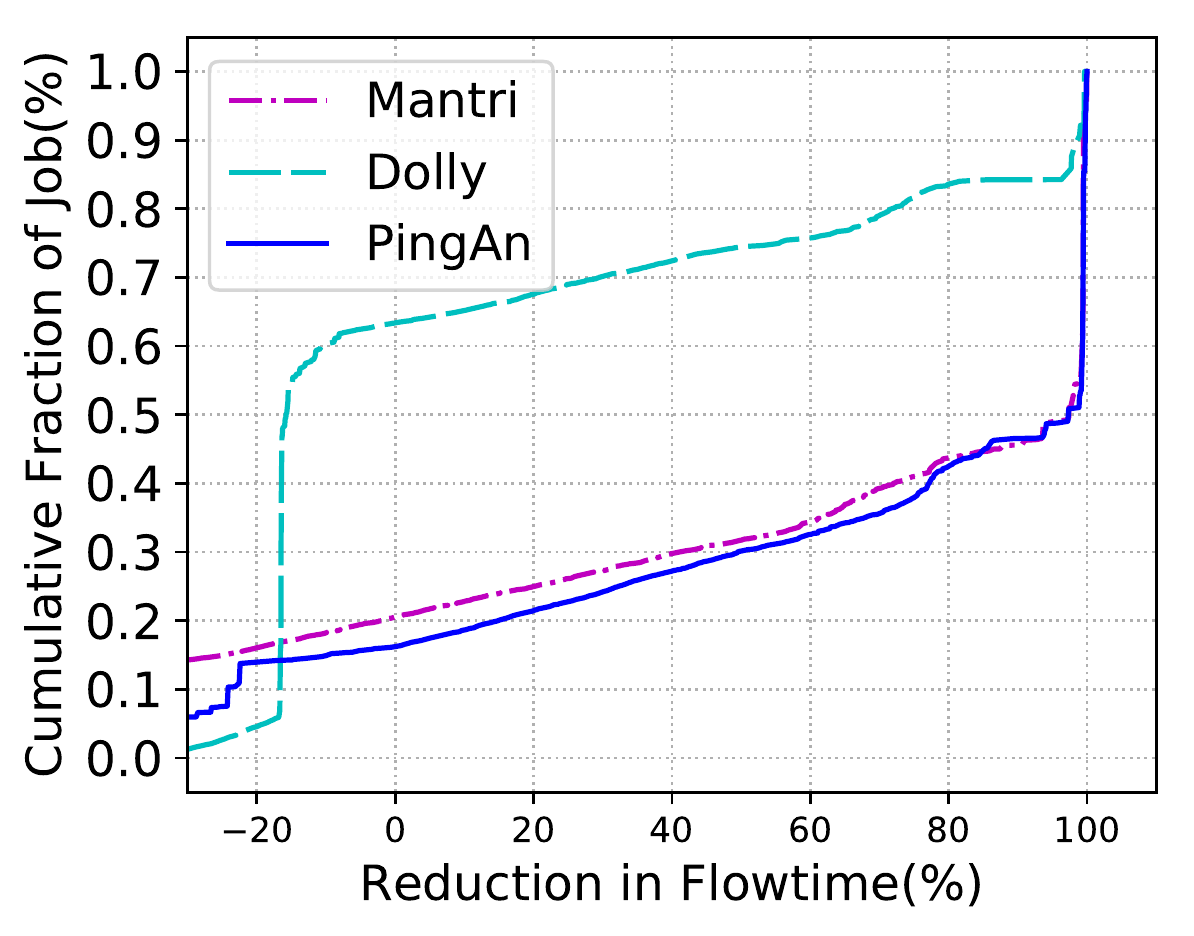}}
  \caption{The details of job performance in PingAn and baselines in different load condition. (a)(c)(e) depicts the CDF of job flowtimes for PingAn and each baselines under three load; (b)(d)(f) depicts the CDF of jobs flowtime reduction ratio to the Flutter for PingAn, Mantri and Dolly under three load.}
  \label{fig:flowtime_comp_under_different_load}
\end{figure}

%


In lightly loaded case, Mantri and Dolly are well-matched.
The sufficient idle slots admits Dolly making enough clones to avoid stragglers, and co-existed task number is little enough for Mantri to detect the straggler and copy it quickly.
As shown in Figure \ref{fig:c1}, the fraction of jobs finishing within 100 seconds is 70.5$\%$ in Dolly and 73.7$\%$ in Mantri.
However, PingAn performs better that 76.9$\%$ jobs finishes within 100 seconds under light load.
PingAn makes copies toward execution efficiency and reliability directly and improves the performance the most.
Seen in another light, as shown in the Figure \ref{fig:r1}, more than 70$\%$ jobs in PingAn has at least 91.4$\%$ reduction in flowtime. In comparison, Mantri and Dolly has only 74.3$\%$ and 85.2$\%$ at their $30^{th}$ reduction ratio.

In moderately loaded case, Mantri is insensitive to promote tasks with a relatively moderate latency which are the majority under the medium load and sometimes the restart copy is ineffective due to costly WAN transfer.
Dolly improves the job efficiency via aggressively making copies and works better than Mantri as shown in Figure \ref{fig:c2} and \ref{fig:r2}. In Figure \ref{fig:c2}, Dolly has 67.5$\%$ jobs finishing within 200 seconds and Mantri only has 61.34$\%$. In Figure \ref{fig:r2}, more than 70$\%$ jobs in Dolly has at least 89.7$\%$ flowtime reduction while 88.1$\%$ in Mantri.
PingAn further precedes Dolly since it insures more efficient and reliable copies to tasks suffering higher risk instead of aggressively cloning as Dolly. Thus, 84.03$\%$ jobs in PingAn can finish within 200 seconds and the value at $30^{th}$ reduction ratio is 94.13$\%$.

In the heavy load case, Mantri effectively restrains the overlong tasks.
Figure \ref{fig:c3} illustrates that 59$\%$ jobs in Mantri finishes within 800 seconds under heavy load. The result is better than the 37.6$\%$ in Dolly but worse than the 71.0$\%$ in PingAn.
PingAn speeds up the job flowtimes the most even under the heavy load via optimizing copy effect as depicted in Figure \ref{fig:r3}.
Concretely speaking, PingAn improves the job flowtimes by 49.6$\%$ at $30^{th}$. In contrast, Mantri is 41.1$\%$ and Dolly even makes 63.4$\%$ jobs flowtime be longer due to its reckless preemption.
\subsection{Impact of Insurance Principle}\label{ch:principleExp}
In this subsection, we verify the effect of efficiency-first reliability-aware principle via comparing the job performance after exchanging the insuring principle in the first two round.
\begin{figure}[!htp]
\setlength{\abovecaptionskip}{-2pt}
\setlength{\belowcaptionskip}{-5pt}
  \centering
  \subfigure[Effect comparison of candidate principles in the first two round of PingAn]{
    \label{fig:eff-reli}
    \includegraphics[width=4.0cm,height=3.3cm]{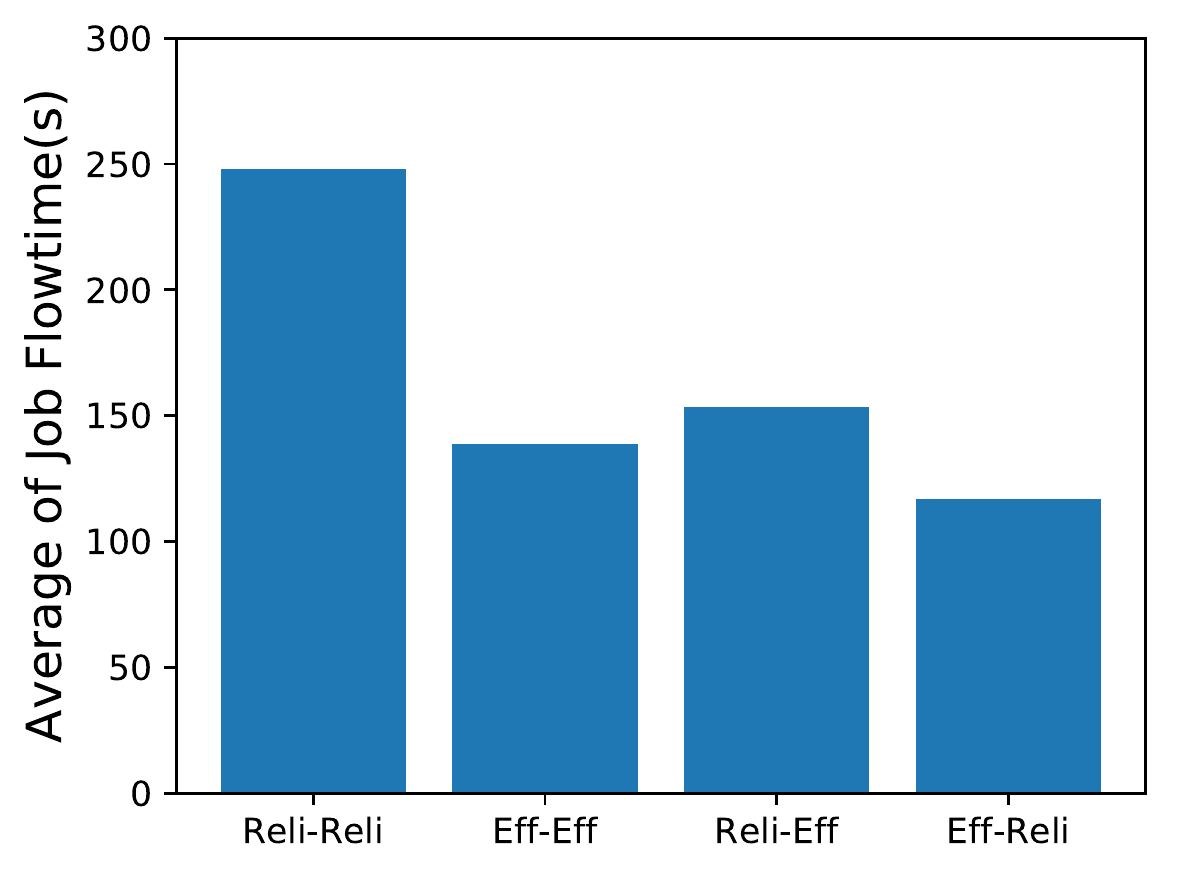}}
 \subfigure[Effect comparison of Job-Greedy and Efficient-first in PingAn]{
    \label{fig:efa-jga}
    \includegraphics[width=4.0cm]{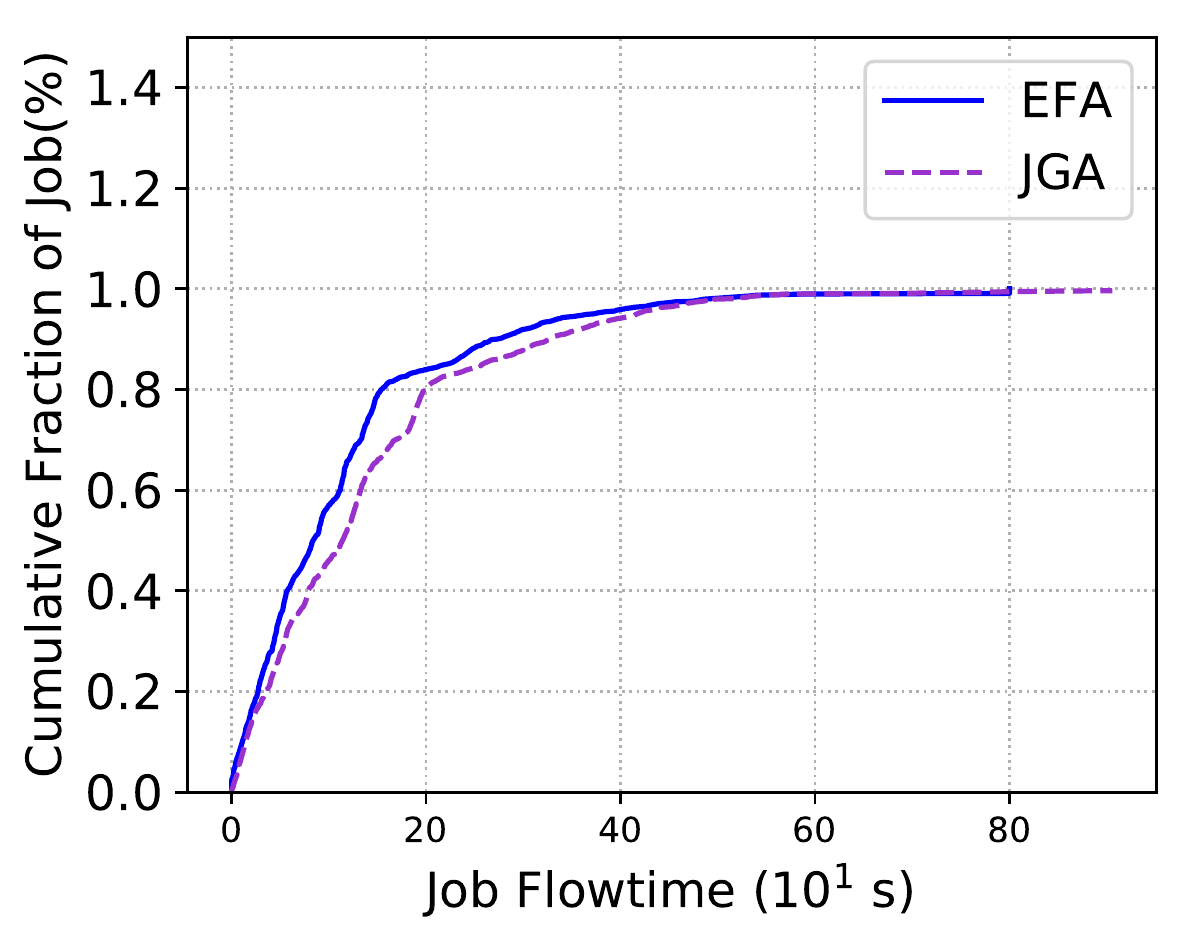}}
  \caption{The effect of efficient-first reliability-aware principle when $\varepsilon=0.6$ in PingAn and load parameter $\lambda=0.07$.}
  \label{fig:efficient-first}
\end{figure}
\vspace{-0.2cm}
The original insuring scheme in PingAn is denoted as $\text{Eff-Reli}$.
For the others, the one that uses reliability-aware in the first round and efficient-first in the second round is denoted as $\text{Reli-Eff}$, the one that uses efficient-first in both two rounds is denoted as $\text{Eff-Eff}$ and the one that uses the reliability-aware in both two rounds is denoted as $\text{Reli-Reli}$.
Figure \ref{fig:eff-reli} shows that $\text{Eff-Reli}$ performs better than the candidates violating the efficiency-first principle and its average job flowtime is less than $\text{Reli-Eff}$ and $\text{Reli-Reli}$ by 18.5$\%$ and 52.8$\%$ respectively.
Although the efficiency is priority to the reliability, however, the reliability is also worthy to consider since $\text{Eff-Eff}$ without the awareness of reliability is worse than $\text{Eff-Reli}$ by 4$\%$ in the average job flowtime.

The efficiency-first principle also works on resource allocation among multiple jobs in the first insuring round.
Figure \ref{fig:efa-jga} indicates that EFA works better than JGA.
Specially, the average flowtime of EFA is less than JGA by 39.4$\%$.
\subsection{Hint on $\varepsilon$ Selection}\label{ch:epsilonselection}

\begin{figure}[htp]
\setlength{\abovecaptionskip}{0pt}
\setlength{\belowcaptionskip}{-5pt}
  \centering
  \subfigure[The average job flowtime variance of $\lambda \in \text{[0.02,0.15]}$ under different $\varepsilon$]{
    \label{fig:e1}
    \includegraphics[width=4.0cm]{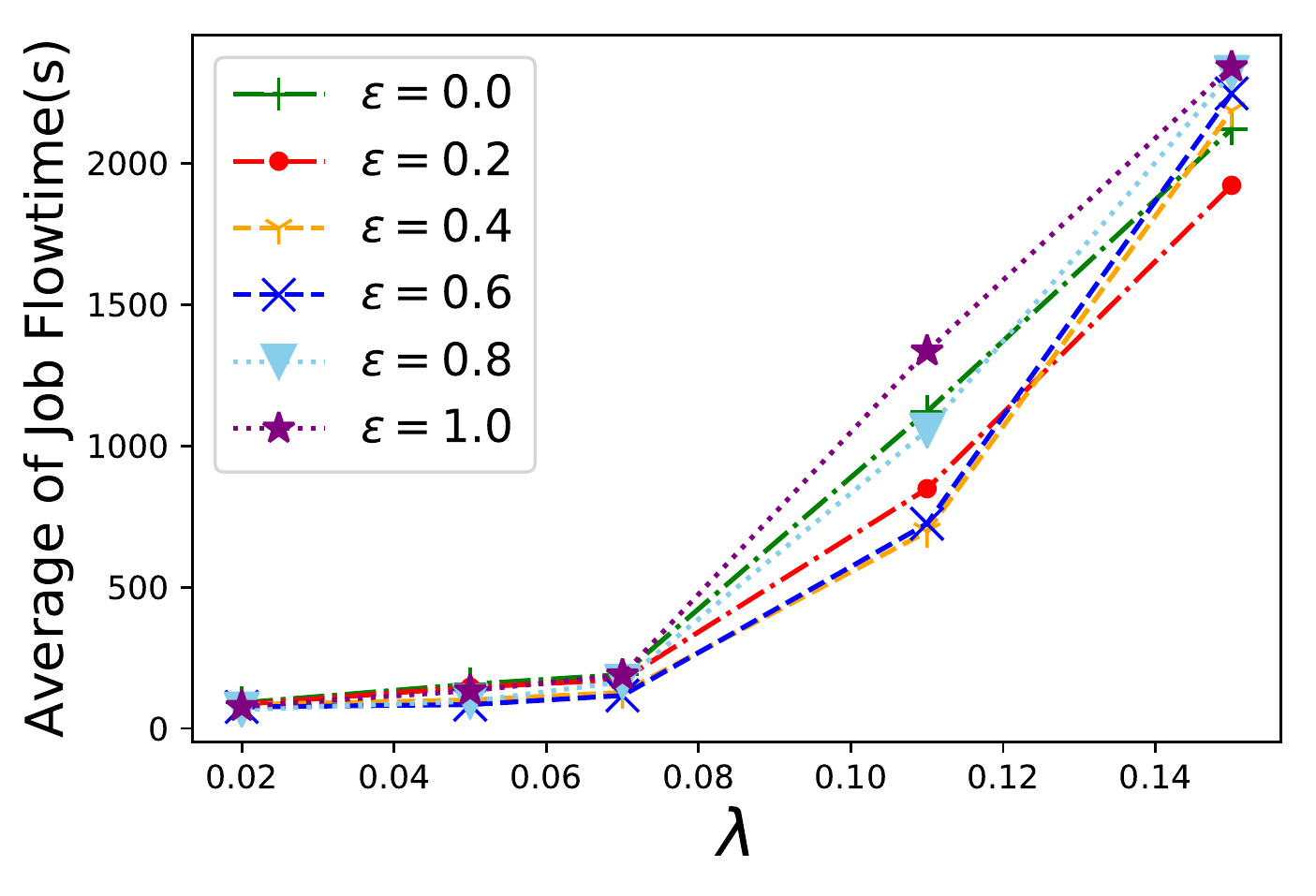}}
  \subfigure[The average job flowtime variance of $\lambda \in \text{[0.02,0.08]}$ under different $\varepsilon$]{
    \label{fig:e2}
    \includegraphics[width=4.0cm]{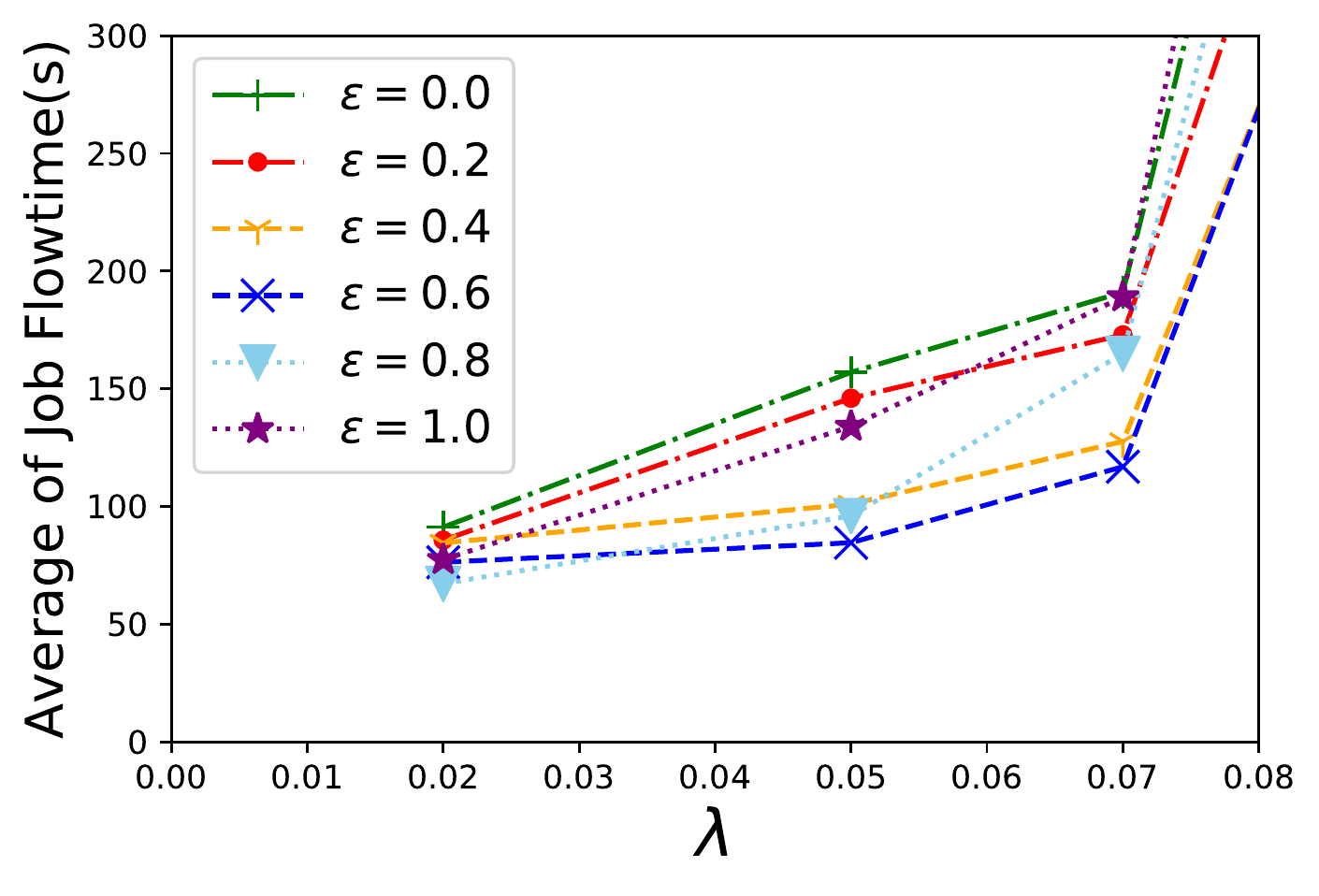}}
  \caption{The relation between $\varepsilon$ and $\lambda$.}
  \label{fig:epsilon_lambda}
\end{figure}
The adjustable performance parameter $\varepsilon$ in PingAn need to tune to some value that best fits the system load condition.
The $\varepsilon$ trades off the overall performance improvement between the acceleration of jobs with smaller workloads and the completion of jobs with larger workloads.
We adjust the Poisson parameters $\lambda$ to control the jobs arriving rate and evaluate the impact of $\varepsilon$ on the average job flowtimes in each load condition.
The evaluation results is depicted in Figure \ref{fig:epsilon_lambda}.

Under five workload arriving rate (let $\lambda$ to be 0.02, 0.05, 0.07, 0.11 and 0.15 respectively), the workload's favourite $\varepsilon$ value is 0.8, 0.6, 0.6, 0.4 and 0.2 respectively.
It can be a hint to select $\varepsilon$ for a system.
For a lightly loaded case, the selection of $\lambda$ is partial to be a moderate or little bigger value to fully utilize idle resources.
For a heavily loaded case, the value prefers to be closer to 0.2 to strive more efficiency for the small jobs arriving at the system.
\section{Conclusion}
In this paper, we focus on an online geo-distributed job flowtimes optimization problem in a cloud-edge system.
To address the unstable and unreliable execution in edges, we propose PingAn insuring algorithm to speed up jobs via inter-cluster task copying and provide a bounded competitive ratio.
PingAn excavates the insuring revenue better on account of the awareness of cluster heterogeneity and costly inter-cluster data fetch on copy execution.
Both of our system implementation and extensive simulation results demonstrate that under any load condition, PingAn can drastically improve the geo-distributed job performance and surpass the best cluster-scale speculation mechanisms by at least 14$\%$.

\bibliographystyle{ACM-Reference-Format}
\bibliography{sample-bibliography}

\appendix
\section{}
\subsection{Proof of Proposition 1}\label{ch:appendix1}
For convenience of analysis, we assumes that the distribution of data processing speed inside a cluster as well as the distribution of data transfer speed inter a cluster-pair can be fit to a continuous distribution.
~\cite{schad2010runtime} supports the assumption. They conducted a performance analysis spanning multiple Amazon EC2 clusters and found several of the performance measurements of VMs - particularly network bandwidth - to be normally distributed.
\begin{proof}
After $n$ insuring rounds in PingAn, The task ${\xi}_l^i$ execution rate $r_{l}^{i}(n)=\mathbb{E}\left [ \max\left \{ V_1,V_2,\cdots,V_n \right \} \right ]$.
The execution rate of the copy $V_x$ follows a distribution, i.e., $V_x\sim Q_x(v)=Pr(V_x<v)$ and let $q_x(v)=Q_x^{'}(v)$.

Let $V_n^r=\max\left \{ V_1,V_2,\cdots,V_n \right \}$ and define $Q_n^r(v)$ as the cumulative distribution function of $V_n^r$.
We have
\begin{sequation}
Q_n^r(v)=\prod\limits_{x=1}^{n}Q_x(v)
\end{sequation}
and further deduce its derivation that
\begin{sequation}\label{eq:q_n^r}
q_n^r(v)={Q_n^r}^'(v)=\sum\limits_{x=1}^{n}{( q_x(v)\cdot \prod\limits_{^{j=1:n;}_{\  j\neq x}}Q_j(v))}
\end{sequation}
In the first place, we prove that
\begin{sequation}\label{eq:first}
(n+1)r_{l}^{i}(n)\geq n\cdot r_{l}^{i}(n+1)
\end{sequation}
when $n\geq 1$.
We expand the left side of Eq. (\ref{eq:first}) as shown in Eq. (\ref{eq:leftside}).
The second equality in Eq. (\ref{eq:leftside}) follows the definition of expectation.
\begin{sequation}\label{eq:leftside}
\begin{split}
(n+1)r_{l}^{i}(n)=(n+1)\mathbb{E}\left [ \max\left \{ V_1,V_2,\cdots,V_n \right \} \right ]\qquad \\
=\!(n\!+\!1)\!\!\int\!\! v\cdot q_{n}^r(v)dv\!=\!n\!\!\int\!\! v\cdot q_n^r(v)dv \!+ \!\!\int\!\! v\cdot q_n^r(v)dv
\end{split}
\end{sequation}
and from the right side, we have
\begin{sequation}\label{eq:rightside}
\begin{split}
n\!\cdot\! r_{l}^{i}(n\!+\!1)\!=\!n\!\cdot \!\mathbb{E}\!\left [ \max\!\left \{ V_1,\!V_2,\!\cdots \!,\!V_{n+1}\! \right \} \right ]\!=\!n\!\!\int \!\!v\!\cdot\! q_{n+1}^r(v)dv \\
=\!n\!\!\int \!\!v\!\cdot \!q_{n}^r(v)Q_{n+1}(v)dv\!+\!n\!\!\int \!\!v\cdot q_{n+1}(v)\prod\limits_{j=1}^n Q_j(v)dv
\end{split}
\end{sequation}
The third equality in the Eq. (\ref{eq:rightside}) applies the definition of $q_{n+1}^r(v)$ in Eq. (\ref{eq:q_n^r}).
Obviously, the first term in the last formula of Eq. (\ref{eq:leftside}) is greater than the first term in the last formula of Eq. (\ref{eq:rightside}) because that $Q_{n+1}(v)\leq 1$.
Consequently, we only need to prove that the remainder of Eq. (\ref{eq:leftside}) and Eq. (\ref{eq:rightside}) satisfies the following inequality.
\begin{sequation}\label{eq:second}
\int v\cdot q_n^r(v)dv \geq n\cdot \int v\cdot q_{n+1}(v)\cdot \prod\limits_{j=1}^n Q_j(v)dv
\end{sequation}
To this end, we unfold the left side in Eq. (\ref{eq:second}) based on the definition of $q_n^r(v)$ and there are
\begin{sequation}\label{eq:leftunfold}
\int v\cdot q_n^r(v)dv = \sum\limits_{x=1}^{n}\int v\cdot q_{x}(v)\prod\limits_{j\neq x}Q_{j}(v)dv
\end{sequation}
Recalling that PingAn greedily insure the best copy for a task in each round. Thus, we have
\begin{sequation}
\mathbb{E}[V_x] \geq \mathbb{E}[V_{n+1}] \quad (x < n+1)
\end{sequation}
Apparently, it follows that
\begin{sequation*}
\begin{split}
\sum\limits_{x=1}^n \mathbb{E}[V_x] &\geq  n\cdot \mathbb{E}[V_{n+1}] \\
\Rightarrow \sum\limits_{x=1}^n \int v\cdot q_{x}(v)dv &\geq n\cdot \int v\cdot q_{n+1}(v)dv \\
\Rightarrow \sum\limits_{x=1}^n \int v\cdot q_{x}(v)\prod\limits_{j\neq x}Q_{j}(v)dv &\geq \sum\limits_{x=1}^n \int v\cdot q_{n+1}(v)\prod\limits_{j\neq x}Q_{j}(v)dv
\end{split}
\end{sequation*}
\begin{sequation}\label{eq:third}
\Rightarrow \sum\limits_{x=1}^n \int v\cdot q_{x}(v)\prod\limits_{j\neq x}Q_{j}(v)dv \geq n\int v q_{n+1}(v)\prod\limits_{j=1}^n Q_{j}(v)dv
\end{sequation}
The fourth inequality in the above follows the factor that $Q_{x}(v)\leq 1$.
Substituting Eq. (\ref{eq:leftunfold}) into Eq. (\ref{eq:third}), we conclude the Eq. (\ref{eq:second}) and further prove the Eq. (\ref{eq:first}).
Based on the Eq. (\ref{eq:first}), for any integer $b\geq a > 0$, we have
\begin{sequation*}
\frac{r_{l}^i(a)}{a} \geq \frac{r_{l}^i(a+1)}{a+1}\geq \cdots \geq  \frac{r_{l}^i(b)}{b}
\end{sequation*}
The proof completes.
\end{proof}
\subsection{Proof of Theorem 3}\label{ch:appendix2}
The potential function we defined in our analysis is extended from \cite{im2016competitively} and \cite{Xu2015Task}.
We assumes that the rate function $s_{l}^{i}(x)$ of the optimal adversary in our potential function analysis is a concave and strictly increasing function of $x$, which is a model generally adopted in many studies about online parallel scheduling problem \cite{Xu2015Task, im2016competitively,edmonds2009scalably,fox2013energy}.
The Proposition \ref{pro:second} is proved in \cite{Xu2015Task}, thus we used it directly in the following theorem proof.
\begin{proposition}\label{pro:second}
Consider any continuous and concave function $f$ : $\mathbb{R}^+\rightarrow \mathbb{R}^+$ with $f(0)\geq 0$. Then for any $b\geq a > 0$, we have $\frac{f(a)}{a}\geq \frac{f(b)}{b}$.
\end{proposition}
\begin{proof}
Let $z_{l}^{i}(t)=\max({d_{l}^{iP}}-{d_{l}^{iO}},0)$ where ${d_{l}^{iP}}$ and ${d_{l}^{iO}}$ represent the remaining unprocessed workload for task $\xi_{l}^{i}$ in Job $J_i$ at time $t$ under the optimal scheduling policy and PingAn insurance algorithm respectively.

The potential function for a single task is defined as follow
\begin{sequation*}
\varphi_{l}^{i}(t)=\frac{z_{l}^{i}(t)}{r_{l}^{i}(M_{\mathcal{K}}/\varepsilon N(t))}
\end{sequation*}
where $M_{\mathcal{K}}=\sum\limits_{k\in \mathcal{K}}M_{k}$.

The overall \textbf{Potential Function} for all the jobs arriving at the system is defined as
\begin{sequation*}
\Psi(t)=\frac{1}{\varepsilon^2}\sum\limits_{J_i \in \eta^{P}(t)}\sum\limits_{\xi_{l}^{i}\in J_{i}}\varphi_{l}^{i}(t)
\end{sequation*}
where $\eta^{P}(t)$ denotes the set of alive jobs in PingAn at time $t$. Further let $\eta^{O}(t)$ and $\eta^{O}_{l}(t)$ indicates the jobs and tasks that have not completed at time $t$ in the optimal scheduling.

The potential function is differentiable, and we have
\begin{sequation}\label{eq:decomposetotalchange}
\mathbb{E}\left [ \frac{d\Psi(t)}{dt} \right ]=\frac{1}{\varepsilon^2}\sum\limits_{J_i\in \eta^{P}(t)}\sum\limits_{\xi_{l}^{i}\in J_{i}}\mathbb{E}\left [ \frac{d\varphi_{l}^{i}(t)}{dt} \right ]
\end{sequation}

Obviously, it holds that $\Psi(0)=\Psi(\infty )=0$ and the value of the potential function does not increase when a job arrives or completes in PingAn and the optimal adversary.
Thus, we analyze the change of $\Psi(t)$ at the time $t$ that no job arrives or completes.
Some notations shown in Table \ref{tab:notation} are used in the following $\Psi(t)$ change analysis.
\begin{table}[!htbp]
\caption{The Notation in Approximation Analysis}
\small
\centering
\begin{tabular}{p{1.2cm}|p{7cm}<{\centering}}
\hline
Notations&Corresponding meaning \\
\hline
\end{tabular}

\begin{tabular}{p{1.2cm}<{\centering}|p{7cm}}
\hline
$f_{i}^{C}$&$C=\left\{ P \text{ for PingAn; }OPT \text{ for optimal} \right\}$ \\
&The completion time of job $J_{i}$ under $C$ algorithm \\
\hline
$C_{i}(t)$&$=\min(f_{i}^{C},t)-a_{i}$ \\
&The accumulated flow time of job $J_{i}$ at time $t$ under $C$ algorithm \\
\hline
$C(t)$& $=\sum_iC_i(t)$ \\
& The accumulated sum of job flowtimes at time $t$ under $C$ algorithm \\
\hline
$C_i$&$=C_i(f_i^C)=C_{i}(\infty)$ \\
& The flowtime of $J_i$ \\
\hline
$C$&$=\sum_{i}P_{i}(\infty )$ \\
& The sum of job flowtimes under $C$ algorthm\\
\hline
\end{tabular}
\label{tab:notation}
\end{table}
\begin{itemize}
\item
Calculating $\Delta^O(t)$, the changes in $\Psi(t)$ due to the optimal scheduling:

For a task $\xi_l^i$, the change made by the optimal scheduling is denoted as $\Delta_l^{iO}=\frac{\mathrm{d} \mathbb{E}[\varphi_{l}^{i}(t)]}{\mathrm{d} t}$.
Based on the definition of potential function, we expand $\Delta_l^{iO}$ and bound it in the equation below
\begin{sequation}\label{eq:task_delta_opt}
\Delta_{l}^{iO}\leq -\frac{\mathbb{E}[\frac{\mathrm{d}(d_l^{iO}(t)) }{\mathrm{d} t}]}{r_{l}^{i}(M_{\mathcal{K}}/\varepsilon N(t))}
\end{sequation}

Applying the definitions Eq. (\ref{c:c6}) and Eq. (\ref{c:c7}) in optimal scheduling with speed function $s_l^i(x)$, we have
\begin{sequation}\label{eq:changeinspeed}
\begin{split}
\mathbb{E}[e_l^{iO}]&=\mathbb{E}\left [ f_{l}^{iO}-st(\xi_{l}^{i}) \right ]=\mathbb{E}[{D_{l}^{i}}/{s_{l}^{j}(x_{l}^{j})}] \\
&=\mathbb{E}[ \int_{st(\xi_{l}^{i})}^{f_{l}^{i}}\mathrm{d}(d_{l}^{iO}(t)) /{s_{l}^{j}(x_{l}^{j})}]
\end{split}
\end{sequation}
According to Eq. (\ref{eq:changeinspeed}), the following formula is yielded
\begin{sequation}\label{eq:speedformula}
\mathbb{E}\left [ \frac{\mathrm{d}(d_{l}^{iO})}{\mathrm{d}t} \right ]={-s_{l}^{i}(x_{l}^{i})}
\end{sequation}
Let $u_{l}^{iO}$ be the number of slots assigned to task $\xi_{l}^{i}$ of job $J_i$ in the optimal scheduling.
Substituting Eq. (\ref{eq:speedformula}) into Eq. (\ref{eq:task_delta_opt}), we have
\begin{sequation}\label{formula2}
\Delta_{l}^{iO}\leq \frac{s_{l}^{i}(u_{l}^{iO})}{r_{l}^{i}(M_{\mathcal{K}}/\varepsilon N(t))}
\end{sequation}
Recalling that the expected rate of a task running on a slot need to be greater than $\frac{1}{1+\varepsilon}$ fraction of the global optimal rate of the task.
We denote the proportion as $\alpha > \frac{1}{1+\varepsilon}$ for logogram.
Based on the rule of lower limit rate, we have
\begin{sequation}\label{formula3}
r_{l}^{i}(M_{\mathcal{K}}/\varepsilon N(t))\geq {\alpha}V_{opt}\geq {\alpha}s_{l}^{i}(M_{\mathcal{K}}/\varepsilon N(t))
\end{sequation}
Substituting the Eq. (\ref{formula3}) into the Eq. (\ref{formula2}) yields that
\begin{sequation}\label{formula4}
\Delta_{l}^{iO}=\frac{s_{l}^{i}(u_{l}^{iO})}{r_{l}^{i}(M_{\mathcal{K}}/\varepsilon N(t))}\leq \frac{s_{l}^{i}(u_{l}^{iO})}{\alpha s_{l}^{i}(M_{\mathcal{K}}/\varepsilon N(t))}
\end{sequation}
Considering two cases, when $u_{l}^{iO}\leq M_{\mathcal{K}}/\varepsilon N(t)$, we have $\Delta_{l}^{iO}\leq 1/\alpha$ as the result of the monotonic property of $s_l^i(x)$ function; when $u_{l}^{iO}> M_{\mathcal{K}}/\varepsilon N(t)$, based on Proposition \ref{pro:second}, we have $\Delta_{l}^{iO}\leq \frac{u_{l}^{iO}}{\alpha M_{\mathcal{K}}/\varepsilon N(t)}\leq \frac{\varepsilon N(t)u_l^{iO}}{\alpha M}$.
In consequence, it holds that
\begin{sequation}
\Delta_{l}^{iO}\leq 1/\alpha+\frac{\varepsilon N(t)u_l^{iO}}{\alpha M_{\mathcal{K}}}
\end{sequation}
Based on Eq. (\ref{eq:decomposetotalchange}), it follows that
\begin{small}
\begin{align}
&\Delta^O(t)=\frac{1}{\varepsilon^2}\sum\limits_{J_i\in \eta^P(t)\cap \eta^O(t)}\sum\limits_{\xi_l^i\in J_i}\Delta_{l}^{iO}(t) \qquad \qquad \\
&\leq \frac{1}{\alpha \varepsilon^2}\sum\limits_{J_i\in\eta^O(t)}\sum\limits_{\xi_l^i \in J_i}1+\frac{N(t)}{\alpha\varepsilon M_{\mathcal{K}}}\sum\limits_{J_i\in\eta^P(t)}\sum\limits_{\xi_l^i \in J_i}u_l^{iO} \\
&\leq \frac{C}{\alpha \varepsilon^2}\sum\limits_{J_i\in\eta^O(t)}1+\frac{N(t)}{\alpha\varepsilon M_{\mathcal{K}}}M_{\mathcal{K}} \label{eq:optimallastbutone} \\
&\leq \frac{C}{\alpha \varepsilon^2}\!\!\sum\limits_{J_i\in\eta^O(t)}\!\!\!\!\mathbb{E}[\frac{\mathrm{d}{OPT}_i(t) }{\mathrm{d} t}]\!+\!\frac{1}{\alpha\varepsilon }\!\!\sum\limits_{J_i\in\eta^P(t)}\!\!\!\!\mathbb{E}[\frac{\mathrm{d}{P}_i(t) }{\mathrm{d} t}] \label{eq:optimallastbutone2} \\
&= \frac{C}{\alpha \varepsilon^2}\mathbb{E}[\frac{\mathrm{d}{OPT}(t) }{\mathrm{d} t}]+\frac{1}{\alpha\varepsilon }\mathbb{E}[\frac{\mathrm{d}{P}(t) }{\mathrm{d} t}] \label{eq:optimallastbutone3}
\end{align}
\end{small}
where $C$ introduced in Eq. (\ref{eq:optimallastbutone}) is the most copy numbers of tasks made in the optimal scheduling.
The second term in Eq. (\ref{eq:optimallastbutone}) is deduced following the slot number restriction.
Based on the flowtime definition, we derive the Eq. (\ref{eq:optimallastbutone2}) and the Eq. (\ref{eq:optimallastbutone3}).

At this point, we get the change bound caused by the optimal scheduling
\begin{sequation}\label{eq:optimalchange}
\Delta^{O}(t)\leq \frac{C}{\alpha \varepsilon^2}\cdot \mathbb{E}[\frac{\mathrm{d}{OPT}(t)}{\mathrm{d}t}]+\frac{1}{\alpha \varepsilon}\mathbb{E}[\frac{\mathrm{d}P(t)}{\mathrm{d}t}]
\end{sequation}
\item
Calculating $\Delta^P(t)$, the changes in $\Psi(t)$ due to PingAn insurance:

PingAn runs at speed of $1+\varepsilon$ faster. Let $u_{l}^{iP}$ be the number of slots assigned to task $\xi_{l}^{i}$ of job $J_i$ in PingAn at time $t$.
We expand $\Delta_l^{iP}$ based on the definitions of potential function and it holds that
\begin{sequation*}
\begin{split}
\Delta^P(t)=\frac{1+\varepsilon}{\varepsilon^2}\sum\limits_{J_i\in \eta^P(t)\cap \eta^O(t)}\sum\limits_{\xi_l^i\in J_i}\frac{\mathbb{E}[\frac{\mathrm{d}(d_l^{iP}(t)) }{\mathrm{d} t}]-\mathbb{E}[\frac{\mathrm{d}(d_l^{iO}(t)) }{\mathrm{d} t}]}{r_{l}^{i}(M_{\mathcal{K}}/\varepsilon N(t))} \\
\leq \frac{1+\varepsilon}{\varepsilon^2}\sum\limits_{J_i\in \eta^P(t)}\sum\limits_{\xi_l^i\in J_i \cap \xi_l^i\notin\eta_l^O(t)}\frac{\mathbb{E}[\frac{\mathrm{d}(d_l^{iP}(t)) }{\mathrm{d} t}]}{r_{l}^{i}(M_{\mathcal{K}}/\varepsilon N(t))} \qquad \qquad \
\end{split}
\end{sequation*}
\begin{sequation}\label{eq:midinequality}
=-\frac{1+\varepsilon}{\varepsilon^2}\!\!\!\!\sum\limits_{J_i\in \eta^P(t)}\sum\limits_{\xi_l^i\in J_i \cap \xi_l^i\notin\eta_l^O(t)}\!\!\frac{r_l^i(u_l^{iP})}{r_{l}^{i}(M_{\mathcal{K}}/\varepsilon N(t))}
\end{sequation}
The third equality is based on the Eq. (\ref{eq:speedformula}) with a replacement of rate $r_l^i(x)$ in PingAn.
According to the insuring policy in PingAn, we have $\sum_l u_l^{iP}\leq h_i(t)=\frac{M_{\mathcal{K}}}{\varepsilon N(t)}$ and $\sum\limits_{J_i\in\eta^P(t)}h_i(t)=M_{\mathcal{K}}$, further based on Proposition \ref{proposition1}, we derive the following Eq. (\ref{eq:midinequality}) that
\begin{sequation*}
\begin{split}
\Delta^P(t)&\leq-\frac{1+\varepsilon}{\varepsilon^2}\sum\limits_{J_i\in \eta^P(t)}\sum\limits_{\xi_l^i\in J_i \cap \xi_l^i\notin\eta_l^O(t)}\frac{r_l^i(u_l^{iP})}{r_{l}^{i}(M_{\mathcal{K}}/\varepsilon N(t))} \\
&\leq -\frac{(1+\varepsilon)N(t)}{\varepsilon}\sum\limits_{J_i\in \eta^P(t)}\sum\limits_{\xi_l^i\in J_i \cap \xi_l^i\notin\eta_l^O(t)}\frac{u_l^{iP}}{M_{\mathcal{K}}} \\
&\leq -\frac{(1+\varepsilon)N(t)}{\varepsilon}(\frac{\sum\limits_{J_i\in \eta^P(t)}\sum\limits_{\xi_l^i\in J_i}u_l^{iP}}{M_{\mathcal{K}}}-\!\!\!\!\!\sum\limits_{\xi_l^i\in\eta_l^O(t)}\frac{u_l^{iP}}{M_{\mathcal{K}}}) \\
&\leq -\frac{(1+\varepsilon)N(t)}{\varepsilon}(\frac{\sum\limits_{J_i\in \eta^P(t)}h_i(t)}{M_{\mathcal{K}}}-\sum\limits_{\xi_l^i\in\eta_l^O(t)}\frac{h_i(t)}{M_{\mathcal{K}}}) \\
&\leq -\frac{(1+\varepsilon)N(t)}{\varepsilon}(\frac{M_{\mathcal{K}}}{M_{\mathcal{K}}}-\frac{\sum\limits_{J_i\in\eta^O(t)}\frac{M_{\mathcal{K}}}{\varepsilon N(t)}}{M_{\mathcal{K}}}) \\
&=-\frac{(1+\varepsilon)N(t)}{\varepsilon}+\frac{1+\varepsilon}{\varepsilon^2}\sum\limits_{J_i\in\eta^O(t)}1 \\
&=-\frac{1+\varepsilon}{\varepsilon}\mathbb{E}[\frac{\mathrm{d}{P}(t) }{\mathrm{d} t}] +\frac{1+\varepsilon}{\varepsilon^2}\mathbb{E}[\frac{\mathrm{d}{OPT}(t) }{\mathrm{d} t}]
\end{split}
\end{sequation*}
At this point, we get the change bound caused by the PingAn insuring
\begin{sequation}\label{eq:pinganchange}
\Delta^P(t)\leq -\frac{1+\varepsilon}{\varepsilon}\mathbb{E}[\frac{\mathrm{d}{P}(t) }{\mathrm{d} t}] +\frac{1+\varepsilon}{\varepsilon^2}\mathbb{E}[\frac{\mathrm{d}{OPT}(t) }{\mathrm{d} t}]
\end{sequation}

\end{itemize}

We integrate the results derived above over time and complete the potential function analysis.
Due to the facts that $\int_{0}^{\infty }\mathbb{E}\left [ \frac{\mathrm{d}\Psi(t)}{\mathrm{d}t} \right ]\mathrm{d}t=\mathbb{E}[\Psi(\infty)]-\mathbb{E}[\Psi(0)]=0$ and $\int_{0}^{\infty }\mathbb{E}\left [ \frac{\mathrm{d}\Psi(t)}{\mathrm{d}t} \right ]\mathrm{d}t\leq \int_{0}^{\infty }(\Delta^O(t)+\Delta^P(t))\mathrm{d}t$, we have
\begin{sequation}\label{eq:totalchange}
-\int_{0}^{\infty }\Delta^P(t)\mathrm{d}t \leq \int_{0}^{\infty }\Delta^O(t)\mathrm{d}t
\end{sequation}
Substituting the Eq. (\ref{eq:optimalchange}) and the Eq. (\ref{eq:pinganchange}) into the Eq. (\ref{eq:totalchange}), it follows that
\begin{sequation}
\begin{split}
\!\!\!\!\!\!\int_{0}^{\infty }(\frac{1+\varepsilon}{\varepsilon}\mathbb{E}[\frac{\mathrm{d}{P}(t) }{\mathrm{d} t}] -\frac{1+\varepsilon}{\varepsilon^2}\mathbb{E}[\frac{\mathrm{d}{OPT}(t) }{\mathrm{d} t}])\mathrm{d}t \qquad \qquad \qquad \\
\quad \leq \int_{0}^{\infty }(\frac{C}{\alpha \varepsilon^2}\mathbb{E}[\frac{\mathrm{d}{OPT}(t)}{\mathrm{d}t}]+\frac{1}{\alpha \varepsilon}\mathbb{E}[\frac{\mathrm{d}P(t)}{\mathrm{d}t}])\mathrm{d}t
\end{split}
\end{sequation}
\begin{sequation}\label{eq:coefficient}
\begin{split}
\Rightarrow \quad \!\!\!\!\!\!\frac{\alpha (1+\varepsilon )\!\!-\!\!1}{\alpha \varepsilon}\!\!\int_{0}^{\infty }\!\!\!\!\mathbb{E}[\frac{\mathrm{d}{P}(t) }{\mathrm{d} t}] \mathrm{d}t \!\leq \!\frac{\alpha(1+\varepsilon)\!+\!C}{\alpha \varepsilon^2}\!\!\int_{0}^{\infty }\!\!\!\!\mathbb{E}[\frac{\mathrm{d}{OPT}(t)}{\mathrm{d}t}]\mathrm{d}t
\end{split}
\end{sequation}
\begin{sequation}
\Rightarrow \int_{0}^{\infty }\mathbb{E}[\frac{\mathrm{d}{P}(t) }{\mathrm{d} t}] \mathrm{d}t\leq \frac{\alpha(1+\varepsilon)+C}{\alpha\varepsilon^2+\alpha\varepsilon-\varepsilon}\int_{0}^{\infty }\mathbb{E}[\frac{\mathrm{d}{OPT}(t)}{\mathrm{d}t}]\mathrm{d}t
\end{sequation}
\begin{sequation}
\Rightarrow \mathbb{E}[P]\leq \frac{\alpha(1+\varepsilon)+C}{\alpha\varepsilon^2+(\alpha-1)\varepsilon}\mathbb{E}[OPT]
\end{sequation}
In the Eq. (\ref{eq:coefficient}), the coefficient of the left term $\frac{\alpha (1+\varepsilon )-1}{\alpha \varepsilon}>0$ as $\alpha > \frac{1}{1+\varepsilon}$ in PingAn, thus the coefficient can divide the right term.

The proof completes.
\end{proof}
\end{document}